\documentclass[12pt]{article}

\usepackage[T1]{fontenc}
\usepackage[utf8]{inputenc}

\usepackage{amsmath,amsfonts,amssymb,amsthm,mathtools}
\usepackage{dsfont,bm}

\usepackage[semicolon,round,longnamesfirst]{natbib}

\usepackage[kerning=true,protrusion=true,expansion]{microtype}
\usepackage[a4paper,left=1.25in,right=1.25in,bottom=1.25in,top=1.25in]{geometry}
\usepackage{setspace}
\usepackage{comment}

\usepackage{sectsty}
\usepackage{titlesec}
\usepackage[title,titletoc,toc,page]{appendix}

\usepackage[usenames,dvipsnames]{xcolor}
\definecolor{mBlue}{HTML}{002FA7}
\definecolor{mGreen}{HTML}{009B55}
\definecolor{mOrange}{HTML}{FF4F00}
\definecolor{mBlack}{HTML}{004242}

\usepackage{pgfplots,tikz,tikz-3dplot}
\pgfplotsset{compat=newest}
\usepgfplotslibrary{fillbetween}
\pgfplotsset{soldot/.style={color=blue,only marks,mark=*}}
\usetikzlibrary{positioning,decorations.markings,arrows,arrows.meta}
\usepgfplotslibrary{patchplots,ternary}

\usepackage{graphics,caption,subcaption}

\usepackage{enumerate}
\usepackage[shortlabels]{enumitem}
\usepackage{accents}
\usepackage{url}
\usepackage{breakcites}
\usepackage{footmisc}
\usepackage[backref=page,breaklinks, pdftex]{hyperref}
\usepackage[nameinlink,noabbrev,capitalise]{cleveref}
\usepackage{theoremref}

\usepackage{algorithm}
\usepackage{algpseudocodex}

\titleformat{\section}[hang]{\normalfont\scshape\large}{\thesection.}{1em}{\centering}
\titleformat{\subsection}[hang]{\normalfont\itshape\bf}{\thesubsection.}{1em}{\centering}
\titleformat{\subsubsection}[hang]{\normalfont\itshape}{\thesubsubsection.}{1em}{\centering}
\let\originalparagraph\paragraph
\renewcommand{\paragraph}[2][.]{\originalparagraph{#2#1}}

\hypersetup{
	backref=true, 
	pagebackref=true, 
	hyperindex=true, 
	colorlinks=true, 
	breaklinks=true, 
	urlcolor=mBlue, 
	linkcolor=mBlue,  
	citecolor=NavyBlue 
	}

\newtheoremstyle{mystyle}
  {}
  {}
  {\itshape}
  {}
  {\bfseries}
  {.}
  { }
  {\thmname{#1}\thmnumber{ #2}\thmnote{ (#3)}}

\theoremstyle{mystyle}
\newtheorem{lemma}{Lemma}
\newtheorem{proposition}{Proposition}
\newtheorem{theorem}{Theorem}
\newtheorem{corollary}{Corollary}

\makeatletter
\newcounter{parenttheorem}

\makeatother

\newtheoremstyle{mydefstyle}
  {}
  {}
  {}
  {}
  {\bfseries}
  {.}
  { }
  {\thmname{#1}\thmnumber{ #2}\thmnote{ (#3)}}

\theoremstyle{mydefstyle}
\newtheorem{definition}{Definition}
\newtheorem{example}{Example}

\theoremstyle{remark}

\newcommand{\placepoints}[7]{
    \foreach \x in {#1,...,#2} {
        \filldraw[#5, opacity=#6] (\x,#3) circle (#4);
    }
}

\makeatletter
\def\mathcolor#1#{\@mathcolor{#1}}
\def\@mathcolor#1#2#3{%
  \protect\leavevmode
  \begingroup
    \color#1{#2}#3%
  \endgroup
}
\makeatother

\DeclareMathOperator{\supp}{supp}

\DeclareMathOperator{\argmax}{argmax}

\DeclareMathOperator{\Ind}{\mathds{1}}

\DeclareMathOperator{\cone}{cone}

\title{\textbf{Redistribution Through Market Segmentation}\thanks{The \textcircled{r} symbol means that the authors' names are presented in a certified random order \citep{Ray2018}. We especially thank Franz Dietrich, Jeanne Hagenbach, Eduardo Perez-Richet and Jean-Marc Tallon for their guidance. We also thank Sarah Auster, Francis Bloch, Gregorio Curello, Francesc Dilmé, Théo Durandard, Piotr Dworczak, Emeric Henry, Frédéric Koessler, Stephan Lauermann, Raphaël Levy, Paola Manzini, Andrei Matveenko, Laurent Mathevet, Meg Meyer, Benny Moldovanu, Paula Onuchic, Pietro Ortoleva, Franz Ostrizek, Rosina Rodriguez Olivera, Vasiliki Skreta, Randolph Sloof, Ludvig Sinander, Colin Stewart, Nikhil Vellodi, and Adrien Vigier for insightful discussions and comments at various stages of the project as well as participants at various seminars and conferences. The project has received funding from the European Research Council (ERC) under the European Union's Horizon 2020 research and innovation programme (101001694---IMEDMC, and 101040122---IMD).}}
\author{Victor Augias \textcircled{r} Alexis Ghersengorin \textcircled{r} Daniel M.A.~Barreto\thanks{\emph{Augias}: University of Bonn, Department of Economics---e-mail: \href{mailto:vaugias@uni-bonn.de}{\texttt{vaugias@uni-bonn.de}}; \emph{Ghersengorin}: University of Bonn, Department of Economics---email: \href{mailto:aghersen@uni-bonn.de}{\texttt{aghersen@uni-bonn.de}}; \emph{Barreto}: Amsterdam School of Economics, University of Amsterdam---e-mail: \href{mailto:d.m.a.barreto@uva.nl}{\texttt{d.m.a.barreto@uva.nl}}.}}
\date{\today}

\begin{document}

\maketitle

\begin{abstract}
    We study how to optimally segment monopolistic markets with a redistributive objective. We characterize optimal redistributive segmentations and show that they (i) induce the seller to price progressively, i.e., charge richer consumers higher prices than poorer ones, and (ii) may not maximize consumer surplus, instead granting extra profits to the monopolist. We further show that optimal redistributive segmentations are implementable via price-based regulation.
\end{abstract}

\section{Introduction}


Firms and online platforms use rich consumer data to segment markets and personalize prices.\footnote{Third-degree price discrimination is often implemented through targeted discounts. For instance, this is how \href{https://sellercentral.amazon.com/help/hub/reference/external/GFM3F4GG5EYCC5XC}{Amazon Brand Tailored promotions} and \href{https://support.google.com/google-ads/answer/7367521?hl=en}{Google promotion assets} operate. See \cite{mohammed17how,wallheimer18are,oecd2022personalised} for more details and evidence.} While these practices can improve economic efficiency and even increase total consumer surplus, questions remain about their \emph{distributional} consequences. 

Distributive concerns have long shaped both policy practice and fairness norms regarding price discrimination. From ``lifeline'' tariffs to student discounts at cultural activities or food stamps to low-income households, price discrimination has repeatedly been justified on the basis of \emph{who} it benefits. Recent regulatory debates on online personalized pricing have expressed similar concerns \citep{WhiteHouse2015,bourreau2018regulation}. Despite this, most of the theoretical literature on price discrimination has focused on studying its impacts on aggregate measures such as output or total consumer surplus. Price discrimination inherently affects consumers differently, and evaluating its welfare effects solely through such aggregates conceals these heterogeneous impacts.


In particular, segmenting markets to maximize total consumer surplus can favor wealthier consumers. As \cite{bbm15} show, monopolistic markets can be segmented so that the whole generated surplus accrues to consumers. This outcome is achieved by exploiting \emph{pooling externalities} \citep{galperti2023}: consumers with a low willingness to pay exert a positive externality on consumers with a high willingness to pay who belong to the same segment by inducing the seller to charge prices below the uniform monopoly price. Adequately pooling consumers in the same segments therefore increases the total share of the surplus going to consumers, particularly those most willing to pay. However, a low willingness to pay is often indicative of adverse economic and social circumstances. Therefore, maximizing consumer surplus can benefit wealthy buyers to the detriment of poorer ones.

This paper aims to better understand these distributive effects of market segmentation. Specifically, we address the normative question of how to segment a market to benefit consumers while prioritizing the welfare of those with \emph{low willingness to pay}. Our two main economic insights are the following. First, unlike segmentations that maximize total consumer surplus, which assign consumers with a high willingness to pay to low-price segments, optimal redistributive segmentations are \emph{positive-assortative}: they assign buyers with a high willingness to pay to high-price segments and buyers with low willingness to pay to low-price segments. As a result, they induce the monopolist to price \emph{progressively}, meaning that wealthier consumers pay higher prices than poorer ones (\cref{sec:monotone}). Second, this \emph{separation} of wealthy and poor consumers in different segments sometimes allows the seller to price discriminate more effectively, and therefore increase its profit compared to uniform pricing. Hence, redistributive segmentations may not maximize total consumer surplus (\cref{sec:rents}).

Our model considers a continuous population of buyers, each having a unit demand for a good sold by a monopolistic seller. Every buyer is characterized by her willingness to pay for the good---her type. We refer to the distribution of types in the economy as the \emph{aggregate market}. Without loss of generality, we assume that the seller's cost of production is zero, so gains from trade are always positive. A market designer can \emph{flexibly} segment the aggregate market, anticipating that the seller will charge a profit-maximizing price in each segment. 

The designer's preferences are represented by a \emph{welfare function}, which describes the social value of a given type consuming at a given price. We focus on welfare functions that align with consumer surplus: welfare equals zero when buyers do not purchase, and is decreasing in price when buyers purchase. Within this class, we study \emph{redistributive} welfare functions, for which the welfare gain resulting from reducing the price paid by buyers with low type is greater than the welfare gain resulting from an equal price reduction for buyers with high type.\footnote{Mathematically, this corresponds to the welfare function being supermodular.} This is motivated by the idea that buyers' willingness to pay reflects both their taste for the good and their \emph{opportunity cost of money}. Prioritizing low willingness to pay buyers thus appropriately captures redistributive concerns when heterogeneity in willingness to pay primarily reflects differences in \emph{ability to pay}: wealthier individuals can typically pay more for what they need.\footnote{We refer to \cite{dworczak21,akbarpour2024redistributive,dworczak2024inequality} for in-depth discussions of this argument. \Cref{app:microfoundation} micro-founds this interpretation in our environment.} Examples of redistributive welfare functions include weighted consumer surplus with Pareto weights \emph{decreasing in willingness to pay}, or \emph{increasing concave} transformations of consumer surplus.

Our main result, \cref{thm:opt_red_seg}, characterizes the set of segmentations that are optimal for all (and only) redistributive welfare functions as those that \textit{saturate} a specific set of the seller's incentive constraints. These binding incentive constraints ensure that no marginal reallocations of buyers across segments can be beneficial for every redistributive welfare function. \cref{thm:opt_red_seg} relies on a key result of independent interest. We define a \emph{redistributive order}, which ranks a segmentation as more redistributive than another if it yields higher aggregate welfare for \emph{every} redistributive welfare function. We characterize this order in terms of elementary reallocations of buyers across segments: a segmentation is \emph{more redistributive} than another if and only if there exists a sequence of such elementary reallocations transforming the latter into the former (\cref{prop:red_order}). \Cref{thm:opt_red_seg} follows then from the characterization of the segmentations that are maximal in the redistributive order in terms of their binding incentive constraints.

We leverage \cref{thm:opt_red_seg} to draw key economic implications of optimal redistributive segmentations. We first show the optimality of progressive pricing by proving that any optimal segmentation must assign buyer types to prices in a \emph{positive-assortative} manner. Specifically, \cref{cor:weak_mon} establishes that any optimal segmentation satisfies a \emph{weak} form of monotonicity, which requires that, if a buyer consumes at a given price, there must \emph{exist} a buyer with a greater willingness to pay who consumes at a higher price. This notion of progressive pricing is mild, since it still allows for instances in which some buyer with higher willingness to pay consumes at a price lower than some buyer with lower willingness to pay. We show in \cref{prop:strong_mon} how a strengthening of redistributive concerns\footnote{As shown below, this corresponds mathematically to the welfare function satisfying a strengthened notion of supermodularity (\cref{def:strong_red}).} eliminates these instances, leading to a stronger notion of monotonicity: if a buyer consumes at a given price, then \emph{all} buyers with a (weakly) greater willingness to pay consume at a higher price. Furthermore, we show that optimal strongly monotone segmentations are obtained by following a simple greedy algorithm that iteratively maximizes the welfare of every type starting from the lowest.

Second, while redistribution does not induce an equity-efficiency trade-off (\cref{lem:efficiency} shows that optimal redistributive segmentations are efficient), we show that it may be in conflict with consumer surplus maximization. The monotonic structure of redistributive segmentations implies separation of buyers with low and high willingness to pay into different segments, possibly allowing the seller to price discriminate more effectively than under consumer surplus maximizing segmentations. As a result, the seller might obtain a strictly positive share of the surplus generated by segmenting the market. In such cases, we say that the seller enjoys a \emph{redistributive rent}. Redistributive rents arise especially when optimal segmentations satisfy the strong form of monotonicity described in the previous paragraph. \Cref{thm:rent} characterizes the markets for which strongly redistributive preferences necessarily induce a redistributive rent for the seller. When rents are not necessary, the optimal segmentation is remarkably simple. It only generates two segments: one where the lowest possible price is charged, pooling all types that would not have bought the good in the unsegmented market; and one segment where the uniform monopoly price is charged.

Our contributions do not rely on a literal interpretation of the segmentation problem. The design problem we consider could be understood metaphorically, as emphasized in \cite{bergemann2019}, as an analytical device that allows us to characterize a ``redistributive frontier’’  in the space of segmentations (\cref{thm:opt_red_seg}) and study the properties of segmentations pertaining to this frontier.  \Cref{prop:red_order}, then, provides a way of understanding the ``gap’’ between any given (and potentially exogenous) segmentation and such redistributive frontier. From this perspective, our results shed light on the structure of redistributive segmentations and on the way redistributive objectives interact with the classical standard of consumer-surplus maximization (as discussed in \cref{sec:discussion}). 

Nevertheless, the growing availability of consumer data and the rapid expansion of segmentation practices by actual market participants also invite a literal reading of the segmentation design problem. 
In practice, segmentation is increasingly performed by private actors whose objectives are generally distinct from redistribution or welfare maximization. For example, online platforms can focus on maximizing the number of transactions \citep{Romanyuk2019}, while data intermediaries may aim at maximizing profits of their clients \citep{Yang2022,terstiege26optimal}. In the last section of the paper, we therefore investigate under what conditions redistributive segmentations can be implemented by a regulator who can only imperfectly observe the segmentation chosen by an agent, assumed here to be the seller.\footnote{The result would still hold if the agent had a different objective function, provided it is aligned with efficiency. For example, a platform remunerated by a flat fee on transactions would seek to maximize participation, and thus segment efficiently.} We introduce a notion of \emph{price-based implementation}: a segmentation is implementable through price-based regulation if there exists no alternative segmentation that strictly increases the seller's profit while inducing the \emph{same distribution of prices}.\footnote{This notion of implementation is equivalent to \citeauthor{linliu}'s (\citeyear{linliu}) notion of \emph{credibility}. We further discuss the connection to that contribution in \cref{sec:implementation}, \cref{foot:credibility}.} This notion of implementation is motivated by the fact that although regulators are likely to possess less information than data intermediaries, they can plausibly observe the distribution of charged prices.\footnote{\label{footnote:websites}Many websites, such as Camelcamelcamel, Keepa, Price Before, and Honey, already specialize in tracking the evolution, frequency and extent of price fluctuations for products sold on Amazon.} In the context of regulating online platforms, for instance, it only involves monitoring transaction data and imposing penalties based on the divergence of the empirical price distribution from the desired target distribution. In \cref{prop:implementation}, we show that \emph{any} optimal redistributive segmentation is implementable through price-based regulation.\footnote{In fact, \cref{prop:implementation} states a more general result of independent interest: any \emph{efficient} segmentation can be implemented via price-based regulation.}

\paragraph{Related Literature}

Our paper contributes to the analysis of the welfare effects of third-degree price discrimination, following \cite{pigou} and \cite{robinson}. Much of this literature focuses on studying how given market segmentations affect efficiency, consumer surplus, total output, and prices \citep{schmalensee,varian,Aguirre2010Monopoly,Cowan2016Welfare, rhodeszhou}. 

\cite{legrand} is the first to study theoretically the use of third-degree price discrimination as a redistributive device for public firms. His analysis is expanded by \cite{steinberg} to the case where there is competition between for-profit firms and a (distributionally minded) non-profit provider.

\cite{bbm15} pioneered an information design approach to third-degree price discrimination.\footnote{See \cite{kamenicareview} and \cite{bergemann2019} for comprehensive reviews of the Bayesian persuasion and information design literature.} They show that any pair of consumer surplus and producer profit is attainable through some segmentation if the seller's profit is at least the uniform-pricing level, consumer surplus is non-negative, and total surplus does not exceed the efficient level. \cite{hagpanah_siegel_aer} and \cite{terstiege2023} extend this analysis to multi-product environments, \cite{elliott2024Market} and \cite{bergemann2024alignment} study competitive markets, 
and \cite{Salas2025} examines how optimal segmentation is shaped by endogenous buyer participation. 
We contribute to this literature by studying the impact of market segmentation on consumer welfare beyond the standard of consumer surplus. In particular, we study redistributive welfare functions and show that prioritizing low-type consumers does not hinder efficiency but may not maximize total consumer surplus.

A burgeoning literature rooted in the seminal contribution of \cite{weitzman1977} shows how markets \citep{dworczak21}, allocation mechanisms \citep{Condorelli2013,Reuter2020,akbarpour2024vaccines,akbarpour2024redistributive,kang2023,kang2024a,kang2024b}, or commodity taxation \citep{pai2022taxing,Ahlvik2025,Doligalski2025} can be optimally designed to mitigate economic inequality.\footnote{\cite{dworczak2024inequality} provides an overview of recent contributions to inequality-aware market design.} Our paper contributes to this literature by showing how \emph{market segmentation} (or, equivalently, \emph{information}) can serve as a redistributive device.\footnote{\cite{arya2022rethinking} study redistribution through a different notion of market segmentation, namely the non-tradability of essential and non-essential goods in a general equilibrium framework.} 

Relatedly, \citet{doval_smolin_22} analyze how information provision affects the welfare of heterogeneous agents. They introduce the \emph{Bayes welfare set}, defined as the set of agents’ welfare profiles that can be achieved under some information policy. Their framework applies to a broad class of persuasion environments, including the monopolistic market we study. They characterize this set geometrically and show that its Pareto frontier is generated by the solutions to a family of Bayesian persuasion problems in which a weighted-utilitarian social planner plays the role of the designer. This delivers a complete description of the welfare effects that \emph{any} segmentation can have on consumers of different types.

Rather than characterizing which welfare profiles are feasible, we study how to \emph{optimally} distribute consumer welfare through information design. To do so, we specify a class of redistributive welfare functions---those that are supermodular in the types and the prices---and fully characterize \emph{all} segmentations that are optimal for such objectives, for any distribution of consumer types. This characterization implies the use of \emph{monotone} information structures as optimal policies---structures that are often viewed as especially attractive in the persuasion literature because of their simplicity \citep[see, e.g.,][]{Ivanov2021,Mensch2021,Nikzad2021,Kolotilin2022,Onuchic2023,Kolotilin2024}. This literature often obtains monotone structures as optimal in the case where payoffs are \emph{linear} in the state (that is, only the posterior mean matters).\footnote{Notable exceptions are, e.g., \cite{Mensch2021} and \cite{Nikzad2021}.} This is not the case in the environment we consider---both the sender's and receiver's payoffs are non-linear. We thus develop a new method inspired by the literature on stochastic orders (cf.~\cref{prop:red_order}) to prove our characterization.\footnote{We discuss the connection between our approach and the literature on stochastic orders and optimal transport in greater detail in \cref{sec:red_order}.}

The distributive impacts of price discrimination have also been investigated empirically. \citet{Dube2023} conduct a randomized controlled field experiment to study the welfare implications of personalized pricing. While they find that personalized pricing reduces total consumer surplus, aggregate consumer welfare can increase under inequality-averse welfare criteria, indicating that buyers with low willingness to pay benefit the most from the implemented price discrimination. Similarly, \citet{buchholz2020value} examine price discrimination on a two-sided platform for cab rides and find that, although aggregate consumer surplus falls, more price-sensitive buyers gain from personalized pricing.

\section{Model}

\paragraph{Primitives}

A monopolistic seller (she) supplies a good to a unit mass of buyers. Buyers have a unit demand and differ in their willingness to pay for the good, $\theta$, assumed to belong to the finite set $\Theta=\{\theta_{1},\dots,\theta_{K}\}\subset\mathbb{R}_{+}$, where $\theta_1 < \dots < \theta_K$. We refer to the willingness to pay as the buyer's type and use the two terms interchangeably. A buyer willing to pay $\theta$ consumes at a price $p\in \mathbb{R}_{+}$ if and only if $\theta\geq p$. We assume that the seller's marginal cost of production is zero, so the gains from trade are positive.\footnote{This assumption is without loss of generality. If the seller's marginal cost were equal to $c>0$, then one could redefine the buyer type as $\tilde{\theta}=\max\{0, \theta-c\}$ for all $\theta\in \Theta$.} The \emph{aggregate market} (in short the \emph{market}) $\mu \in \Delta(\Theta)$ specifies the initial distribution of types\footnote{We assume throughout the paper, and without loss of generality, that $\mu(\theta)>0$ for all $\theta\in \Theta$.}.

\paragraph{Market segmentation}

A market \emph{segmentation} is a division of the population of buyers into different market segments indexed by $s\in S$, where $S$ is a sufficiently large finite set. Formally, a segmentation corresponds to a probability distribution $\sigma \in \Delta(\Theta\times S)$ describing the mass $\sigma(\theta,s)$ of type $\theta$ buyers assigned to segment $s$, and satisfying the following \emph{market consistency} constraint:
\begin{equation}\label{eqn:market_consistency}
   \forall \theta\in \Theta, \quad \sum_{s\in S} \sigma(\theta,s)=\mu(\theta). \tag{$\mathrm{C}_{\mu}$}
\end{equation}

Condition \labelcref{eqn:market_consistency} ensures that the mass of type $\theta$ buyers assigned to all segments coincides with the mass of those types in the market.\footnote{We can equivalently think of any pair $(\sigma,S)$ as a \textit{Blackwell experiment} which, for any prior distribution $\mu$ over the state space $\Theta$, describes the joint distribution of signals $s\in S$ and states of the world $\theta\in \Theta$, and satisfies the \emph{Bayes-plausibility constraint} \labelcref{eqn:market_consistency} \citep{Aumann1995,kamenica_gentzkow11}. According to that interpretation, each $\sigma$ represents a way of revealing \emph{information} about the buyers' types to the seller.}

\paragraph{Seller's pricing}

When faced with a market segmentation, the seller sets prices to maximize her profit in each segment. That is, for any $\mu\in\Delta(\Theta)$, any  $\sigma$ satisfying \labelcref{eqn:market_consistency}, and any $s\in S$ such that $\sigma(\theta,s)>0$ for some $\theta\in \Theta$, the seller solves
\begin{equation*}
   \max_{p\in P} \; p \sum_{\theta\geq p} \sigma(\theta,s),
\end{equation*}
where the set of prices is given by $P=\Theta$. Indeed, there is no loss of generality in assuming that the seller charges prices in $\Theta$ since, for any $k\in\{1,\dots,K\}$, charging a price falling strictly between two consecutive types $\theta_{k}$ and $\theta_{k+1}$ always yields a lower profit to the seller than charging either $p = \theta_{k}$ or $p = \theta_{k+1}$.\footnote{Similarly, charging any price lower than $\theta_{1}$ or greater than $\theta_{K}$ is never optimal.} 

\paragraph{Direct and obedient segmentations}

It is without loss of generality to focus on \emph{direct} and \emph{obedient} market segmentations.\footnote{Precisely, it is without loss given the objectives we define below. See \cite{bergemann2019}.} In such segmentations, $S=P$, so that the segments are indexed by a recommended price, and the seller finds it optimal to follow this recommendation. That is:
\begin{equation}\label{eqn:obedience}
    \forall p,q \in P, \quad p\sum_{\theta \geq p}\sigma(\theta,p) \geq q \sum_{\theta \geq q} \sigma(\theta, p). \tag{Ob}
\end{equation}

For any $\mu\in\Delta(\Theta)$, we denote the set of direct and obedient segmentations of $\mu$ as $\Sigma(\mu)$. We henceforth slightly abuse terminology by referring to any element of $\Sigma(\mu)$ as a segmentation of $\mu$. For any $\mu\in \Delta(\Theta)$ and any $\sigma \in \Sigma(\mu)$, we let $\sigma_{P}\in\Delta(P)$ be the \emph{distribution of recommended prices} induced by $\sigma$, defined by:\footnote{As will become clear after \cref{lem:efficiency} is presented, in optimal segmentations the distribution of recommended prices will coincide with the distribution of realized transaction prices.}
\begin{equation*}
    \forall p \in P, \quad \sigma_{P}(p)=\sum_{\theta\in\Theta} \sigma(\theta,p).
\end{equation*}

For any $\mu\in\Delta(\Theta)$, we let $p^{\star}_{\mu}\in P$ be the \emph{uniform optimal price}, that is, the price the seller would optimally charge in the aggregate market:\footnote{Whenever there are several uniform optimal prices for some market $\mu\in\Delta(\Theta)$, we let $p^{\star}_{\mu}$ be the lowest one. This tie-breaking rule is inconsequential for any of our results.}
\begin{equation*}
    \forall p\in \supp(\mu), \quad p^{\star}_{\mu}\sum_{\theta\geq p^{\star}_{\mu}} \mu(\theta) \geq p \sum_{\theta\geq p} \mu(\theta).
\end{equation*}

Let $\mu\in\Delta(\Theta)$. An example of direct segmentation of $\mu$ is illustrated in \cref{fig:direct_seg}.
\begin{figure}[h!]
    \begin{center}
        \begin{tikzpicture}[scale=1, transform shape]
            \draw[step=1cm,mBlack,very thin,opacity=0.2] (0,0) grid (6,6);
            \draw[thick,-stealth] (0,0) -- (7,0) node[anchor=west] {$\theta$};
            \draw[thick,-stealth] (0,0) -- (0,7) node[anchor=south] {$p$};
            
            \node[below left] at (0,0) {$(\theta_1,\theta_1)$};
            \node[below] at (3,0) {$p^{\star}_{\mu}$};
            \node[below] at (6,0) {$\theta_{K}$};
            
            \node[left] at (-0.1,3) {$p^{\star}_{\mu}$};
            \node[left] at (-0.1,6) {$\theta_{K}$};
        
            \node[left] at (-0.1,2) {$p'$};
            \node[left] at (-0.1,4) {$p''$};
            
            
            \placepoints{0}{6}{3}{3pt}{mBlack}{0.4};
            
            \placepoints{0}{2}{2}{3pt}{mBlack}{0.8};
            \placepoints{4}{5}{2}{3pt}{mBlack}{0.8};
            \placepoints{6}{6}{2}{3pt}{mBlack}{0.8};
            
            \placepoints{0}{3}{4}{3pt}{mBlack}{0.8};
            \placepoints{5}{6}{4}{3pt}{mBlack}{0.8};
        
            \draw[-stealth, very thick, mOrange, opacity=0.65] (0,3) -- ++(0,-0.75);
            \draw[-stealth, very thick, mOrange, opacity=0.65] (0,3) -- ++(0, 0.75);
            \draw[-stealth, very thick, mOrange, opacity=0.65] (1,3) -- ++(0,-0.75);
            \draw[-stealth, very thick, mOrange, opacity=0.65] (1,3) -- ++(0,+0.75);
            \draw[-stealth, very thick, mOrange, opacity=0.65] (2,3) -- ++(0,-0.75);
            \draw[-stealth, very thick, mOrange, opacity=0.65] (2,3) -- ++(0,+0.75);
            \draw[-stealth, very thick, mOrange, opacity=0.65] (3,3) -- ++(0,+0.75);
            \draw[-stealth, very thick, mOrange, opacity=0.65] (4,3) -- ++(0,-0.75);
            \draw[-stealth, very thick, mOrange, opacity=0.65] (5,3) -- ++(0,-0.75);
            \draw[-stealth, very thick, mOrange, opacity=0.65] (5,3) -- ++(0,+0.75);
            \draw[-stealth, very thick, mOrange, opacity=0.65] (6,3) -- ++(0,-0.75);
            \draw[-stealth, very thick, mOrange, opacity=0.65] (6,3) -- ++(0,+0.75);
            
        \end{tikzpicture}
    \end{center}
    \caption{A direct segmentation with two segments.}
    \label{fig:direct_seg}
\end{figure}
The initial mass of consumers, located in the row $p^{\star}_{\mu}$ in the grid $\Theta\times P$, is divided between two segments $p'$ and $p''$ vertically, so \labelcref{eqn:market_consistency} is satisfied. Furthermore, \labelcref{eqn:obedience} requires that the mass moved at each point be such that $p'$ and $p''$ are, respectively, profit-maximizing prices in segments $p'$ and $p''$.

\paragraph{Redistributive welfare functions}

We consider a designer (he) who can segment the market flexibly. His preferences over market outcomes are represented by a \emph{welfare function} $w\colon \Theta \times P \to \mathbb{R}$. For any welfare function $w$ and segmentation $\sigma\in\Sigma(\mu)$, the \emph{aggregate welfare} is defined by
\begin{equation*}
    \sum_{(\theta,p)\in \Theta\times P} w(\theta, p) \, \sigma(\theta,p).
\end{equation*}

The value $w(\theta,p)$ thus represents the individual contribution to the aggregate welfare when a buyer of type $\theta$ is assigned to a segment where the seller charges the price $p$. 

We restrict attention to welfare functions satisfying three conditions. First, the value from assigning a buyer to a segment where he does not consume is zero, and the value from assigning a buyer to a segment where he consumes is non-negative. That is, $w(\theta,p)=0$ if $\theta < p$, and $w(\theta,p)\geq 0$ if $\theta \geq p$. Second, assigning any buyer to a segment where the seller charges a lower price always improves welfare. That is, $p\mapsto w(\theta_{k},p)$ is a non-increasing function on $\{\theta_{1}, \dots, \theta_{k} \}$ for all $k\in \{1,\dots,K\}$. The two preceding conditions capture the fact that social welfare aligns with consumer surplus.\footnote{It is straightforward to verify that the \emph{buyer-utilitarian} welfare function, given by $w(\theta,p)=(\theta-p)\mathds{1}_{\theta\geq p}$ for all $(\theta,p)\in\Theta \times P$, satisfies these two minimal conditions.} The third condition captures the designer's \emph{redistributive concerns}. We require that any welfare function be \emph{supermodular} conditional on buyers consuming. That is, for any $(\theta,p),(\theta',p')\in \Theta\times P$ such that $\theta'>\theta \geq p'>p$,
\begin{equation}\label{cond:R}
     w(\theta,p)-w(\theta,p') \geq w(\theta',p)-w(\theta',p'). \tag{R}
\end{equation}

Condition \labelcref{cond:R} ensures that any price discount from $p'$ to $p$ increases welfare more when it is granted to a low-type $\theta$ rather than to a high-type $\theta'$.

We say that a welfare function is \emph{redistributive} if it satisfies the three conditions above, and denote by $\mathcal{R}$ the class of those welfare functions. Furthermore, we say that $w$ is \emph{strictly redistributive} if it satisfies the condition \labelcref{cond:R} strictly on all $\Theta\times P$, and is such that $p\mapsto w(\theta_k,p)$ is strictly decreasing on $\{\theta_{1}, \dots, \theta_k \}$ for all $k\in \{1,\dots,K\}$. We denote the class of strictly redistributive welfare functions by $\overline{\mathcal{R}}$.

In \cref{ex:decreasing_pareto,ex:deacr_marg_surplus} below, we highlight that the class of redistributive welfare functions encompasses two common specifications.
\begin{example}[Decreasing Pareto weights]\label{ex:decreasing_pareto}
Let $\lambda\colon \Theta \to \mathbb{R}$ be such that $\lambda(\theta)\geq 0$ for any $\theta\in \Theta$. The welfare function defined by
\begin{equation}\label{eqn:deacr_Pareto_w}
    \forall(\theta,p)\in\Theta\times P, \quad w(\theta,p)=\lambda(\theta) (\theta-p) \mathds{1}_{\theta\geq p},
\end{equation}
belongs to $\mathcal{R}$ whenever $\theta\mapsto\lambda(\theta)$ is non-increasing, and belongs to $\overline{\mathcal{R}}$ whenever $\theta\mapsto\lambda(\theta)$ is strictly decreasing. The class of welfare functions given by \labelcref{eqn:deacr_Pareto_w} for some decreasing function $\lambda$ corresponds to the family of weighted buyer-utilitarian welfare functions with \emph{decreasing Pareto weights}. Under such welfare functions, the value to the designer of an additional unit of surplus is lower for high-types than for low-types.  Decreasing Pareto weights can endogenously arise under a utilitarian objective if buyers with a low willingness to pay have a higher marginal utility of money than those with a high willingness to pay \citep{Condorelli2013,dworczak21}, or if buyers are identified by unobserved welfare weights that negatively correlate with their willingness to pay \citep{akbarpour2024redistributive}.
\end{example}
\begin{example}[Increasing concave transformation]\label{ex:deacr_marg_surplus}
Let $u\colon\mathbb{R}\to\mathbb{R}$ be an increasing function such that $u(0)=0$. The welfare function defined by
\begin{equation}\label{eqn:concave_trans}
    \forall(\theta,p)\in\Theta\times P, \quad w(\theta,p)=u\bigl((\theta-p)\mathds{1}_{\theta\geq p}\bigr),
\end{equation}
belongs to $\mathcal{R}$ if $u$ is concave, and belongs to $\overline{\mathcal{R}}$ if $u$ is strictly concave. Increasing concave transformations of surplus correspond to the class of welfare functions usually considered in optimal taxation \citep{mirrlees1971exploration} and insurance \citep{Rothschild1976}.
Concavity reflects equity concerns, as it implies diminishing marginal returns to welfare from additional units of surplus. 
\end{example}

Note also that any welfare function which is a combination of \cref{ex:decreasing_pareto,ex:deacr_marg_surplus}, given by
\begin{equation*}
    \forall(\theta,p)\in\Theta\times P, \quad w(\theta,p) = \lambda(\theta) u\bigl( (\theta - p) \mathds{1}_{\theta\geq p} \bigr),
\end{equation*}
belongs to $\mathcal{R}$ whenever $\theta\mapsto\lambda(\theta)$ is sufficiently decreasing, or $u$ is sufficiently concave, or both.

\Cref{app:microfoundation} provides a micro-foundation for these specifications. We show that, even for a purely buyer-utilitarian designer, redistributive welfare functions arise when willingness to pay reflects both taste and income and when higher willingness-to-pay buyers have weakly lower marginal utility for money. Under concave utility for money, this holds, for instance, when income is stochastically increasing in willingness to pay. The appendix also shows that weighted-surplus welfare functions such as those in \cref{ex:decreasing_pareto} can be interpreted as reduced-form buyer-utilitarian objectives, with Pareto weights arising endogenously from income effects rather than being imposed as primitives.


\paragraph{Segmentation design}

For any given welfare function $w \in \mathcal{R}$ and an aggregate market $\mu \in \Delta(\Theta)$, the designer's optimization program consists of choosing a segmentation of $\mu$ to maximize aggregate welfare:
\begin{equation}\label{eqn:designer_pb}
    \max_{\sigma\in\Sigma(\mu)} \; \sum_{(\theta,p)\in \Theta\times P} w(\theta,p) \, \sigma(\theta,p). \tag{$\mathrm{P}_{w,\mu}$}
\end{equation}

\section{An Illustrative Example}\label{sec:illustrative_example}

This section illustrates our main results and economic insights in a three-type example. Think of the traded good as a normal good for which buyers' willingness to pay mainly reflects their ability to pay. We endow the designer with a weighted utilitarian objective and Pareto weights decreasing in buyers' willingness to pay, as in \cref{ex:decreasing_pareto}.


Assume that buyer types are $\Theta = \{1,2,3\}$ and the aggregate market is $\mu = (0.3, 0.4, 0.3)$. The profit-maximizing uniform monopoly price is $p^\star_\mu = 2$. For any $(\theta,p) \in \Theta^2$, welfare is
\begin{equation*}
    w(\theta,p) = \lambda(\theta)\,(\theta - p)\,\Ind_{\theta \geq p},
\end{equation*}
with welfare weights $\lambda(3) = 1 < \lambda(2)$. The value of $\lambda(1)$ is irrelevant, since type-1 consumers never obtain positive surplus (the lowest feasible price is $1$).

Our first result, \cref{lem:efficiency}, shows that we can, without loss, restrict attention to \emph{efficient} segmentations, that is, segmentations under which every buyer type trades (see \cref{sec:efficiency}). Efficient segmentations are exactly the segmentations supported weakly below the diagonal $p=\theta$, corresponding to the shaded region in each panel of \cref{fig:ill_example_2}.

A simple efficient segmentation assigns all type-1 buyers to segment 1 and all type-2 and type-3 buyers to segment 2. \Cref{fig:ill_example_2a} depicts this segmentation, with dot size proportional to buyer mass. This segmentation cannot be optimal. To see this, observe that the designer can implement two reallocations of buyers across segments that strictly improve his objective (and, more generally, weakly improve any redistributive welfare criterion). In particular, he can reassign some type-$2$ and type-$3$ buyers from segment $2$ to segment $1$, as indicated by the downward arrows in \Cref{fig:ill_example_2a}. We refer to such reallocations as \emph{downward transfers} (\cref{def:down_trans}). As long as they respect obedience constraints, downward transfers are always welfare-improving because they reduce the prices faced by some buyers.

\begin{figure}[h]
    \centering
    \begin{subfigure}[c]{0.495\textwidth}
        \centering
        \begin{tikzpicture}[scale=1.25, transform shape]
            \draw[step=1cm,mBlack,very thin,opacity=0.2] (1,1) grid (3,3);
            \draw[thick,-stealth] (1,1) -- (3.2,1) node[anchor=west,scale=1.1] {$\theta$};
            \draw[thick,-stealth] (1,1) -- (1,3.2) node[anchor=south,scale=1.1] {$p$};
            \draw[thick,mBlack,opacity=0.5] (1,1) -- (3,3);
            \fill[mBlack,fill opacity=0.05] (1,1) -- (3,3) -- (3,1) -- cycle;

            \node[anchor=north,scale=0.95] at (1,1) {1};
            \node[anchor=north,scale=0.95] at (2,1) {2};
            \node[anchor=north,scale=0.95] at (3,1) {3};
            \node[anchor=east,scale=0.95]  at (1,1) {1};
            \node[anchor=east,scale=0.95]  at (1,2) {2};
            \node[anchor=east,scale=0.95]  at (1,3) {3};

            \filldraw[mBlack,opacity=0.85] (1,1) circle (3pt); 
            \filldraw[mBlack,opacity=0.85] (2,2) circle (4pt); 
            \filldraw[mBlack,opacity=0.85] (3,2) circle (3pt); 

            \draw[-stealth,thick,mBlue,opacity=1, very thick,opacity=0.55] (3,2) to node[midway, right, very thick,scale=0.8,opacity=0.55] {$\delta$} ++(0,-1);
            \draw[-stealth,thick,mBlue,opacity=1, very thick,opacity=0.55] (2,2) to node[midway, right, very thick,scale=0.8,opacity=0.55] {$\delta$} ++(0,-1);
        \end{tikzpicture}
        \subcaption{}\label{fig:ill_example_2a}
    \end{subfigure}
    \hfill
    \begin{subfigure}[c]{0.495\textwidth}
        \centering
        \begin{tikzpicture}[scale=1.25, transform shape]
            \draw[step=1cm,mBlack,very thin,opacity=0.2] (1,1) grid (3,3);
            \draw[thick,-stealth] (1,1) -- (3.2,1) node[anchor=west,scale=1.1] {$\theta$};
            \draw[thick,-stealth] (1,1) -- (1,3.2) node[anchor=south,scale=1.1] {$p$};
            \draw[thick,mBlack,opacity=0.5] (1,1) -- (3,3);
            \fill[mBlack,fill opacity=0.05] (1,1) -- (3,3) -- (3,1) -- cycle;

            \node[anchor=north,scale=0.95] at (1,1) {1};
            \node[anchor=north,scale=0.95] at (2,1) {2};
            \node[anchor=north,scale=0.95] at (3,1) {3};
            \node[anchor=east,scale=0.95]  at (1,1) {1};
            \node[anchor=east,scale=0.95]  at (1,2) {2};
            \node[anchor=east,scale=0.95]  at (1,3) {3};

            \filldraw[mBlack,opacity=0.85] (1,1) circle (3.0pt); 
            \filldraw[mOrange,opacity=0.85] (2,1) circle (1.5pt); 
            \filldraw[mBlack,opacity=0.85] (3,1) circle (1.5pt); 

            \filldraw[mBlack,opacity=0.85] (2,2) circle (2.5pt); 
            \filldraw[mBlack,opacity=0.85] (3,2) circle (1.5pt); 

            \draw[-stealth, very thick, mOrange,opacity=0.55] (2,2) to node[midway, left, very thick,scale=0.8,opacity=0.55] {$\varepsilon$} ++(0,-0.89);
            \draw[-stealth, very thick, mOrange,opacity=0.55] (3,1) to node[midway, right, very thick,scale=0.8,opacity=0.55] {$\varepsilon$} ++(0,0.89);
        \end{tikzpicture}
        \subcaption{}\label{fig:ill_example_2b}
    \end{subfigure}
    \\
    \begin{subfigure}[c]{0.495\textwidth}
        \centering
        \begin{tikzpicture}[scale=1.25, transform shape]
        \draw[step=1cm,mBlack,very thin,opacity=0.2] (1,1) grid (3,3);
            \draw[thick,-stealth] (1,1) -- (3.2,1) node[anchor=west,scale=1.1] {$\theta$};
            \draw[thick,-stealth] (1,1) -- (1,3.2) node[anchor=south,scale=1.1] {$p$};
            \draw[thick,mBlack,opacity=0.5] (1,1) -- (3,3);
            \fill[mBlack,fill opacity=0.05] (1,1) -- (3,3) -- (3,1) -- cycle;

            \node[anchor=north,scale=0.95] at (1,1) {1};
            \node[anchor=north,scale=0.95] at (2,1) {2};
            \node[anchor=north,scale=0.95] at (3,1) {3};
            \node[anchor=east,scale=0.95]  at (1,1) {1};
            \node[anchor=east,scale=0.95]  at (1,2) {2};
            \node[anchor=east,scale=0.95]  at (1,3) {3};

        \filldraw[mBlack,opacity=0.85] (1,1) circle (3.0pt); 
        \filldraw[mOrange,opacity=0.85] (2,1) circle (2.6pt); 
        \filldraw[mBlack,opacity=0.85] (3,1) circle (1.1pt); 

        \filldraw[mBlack,opacity=0.85] (2,2) circle (1.5pt); 
        \filldraw[mOrange,opacity=0.85] (3,2) circle (2.6pt); 

        \draw[-stealth, very thick, mOrange,opacity=0.55] (2,2) to node[midway, left, very thick,scale=0.8,opacity=0.55] {$\varepsilon$} ++(0,-0.89);
        \draw[-stealth, very thick, mOrange,opacity=0.55] (3,1) to node[midway, right, very thick,scale=0.8,opacity=0.55] {$\varepsilon$} ++(0,0.89);
        \draw[-stealth, very thick, mBlue,opacity=0.55] (3,2) to node[midway, right, very thick,scale=0.8,opacity=0.55] {$3\varepsilon$} ++(0,1);
        \end{tikzpicture}
        \subcaption{}
        \label{fig:ill_example_3a}
    \end{subfigure}
    \hfill
    \begin{subfigure}[c]{0.495\textwidth}
        \centering
        \begin{tikzpicture}[scale=1.25, transform shape]
        \draw[step=1cm,mBlack,very thin,opacity=0.2] (1,1) grid (3,3);
            \draw[thick,-stealth] (1,1) -- (3.2,1) node[anchor=west,scale=1.1] {$\theta$};
            \draw[thick,-stealth] (1,1) -- (1,3.2) node[anchor=south,scale=1.1] {$p$};
            \draw[thick,mBlack,opacity=0.5] (1,1) -- (3,3);
            \fill[mBlack,fill opacity=0.05] (1,1) -- (3,3) -- (3,1) -- cycle;

            \node[anchor=north,scale=0.95] at (1,1) {1};
            \node[anchor=north,scale=0.95] at (2,1) {2};
            \node[anchor=north,scale=0.95] at (3,1) {3};
            \node[anchor=east,scale=0.95]  at (1,1) {1};
            \node[anchor=east,scale=0.95]  at (1,2) {2};
            \node[anchor=east,scale=0.95]  at (1,3) {3};

        \filldraw[mBlack,opacity=0.85] (1,1) circle (3.0pt); 
        \filldraw[mOrange,opacity=0.85] (2,1) circle (3.0pt); 

        \filldraw[mBlack,opacity=0.85] (2,2) circle (1.2pt); 
        \filldraw[mOrange,opacity=0.85] (3,2) circle (2pt); 

        \filldraw[mBlack,opacity=0.85] (3,3) circle (1.2pt); 

        \end{tikzpicture}
        \subcaption{}
        \label{fig:ill_example_3b}
    \end{subfigure}
    \caption{(a): an efficient segmentation, and two beneficial downward transfers leading to (b). (b): a consumer-surplus maximizing segmentation, and a beneficial redistributive transfer leading to (c). (c): a consumer-surplus maximizing and redistributively saturated segmentation, and a redistributive and an upward transfers leading to (d). (d): a strongly monotone and redistributively saturated segmentation. Dot sizes are proportional to masses and orange dots indicate binding obedience constraints.}
    \label{fig:ill_example_2}
\end{figure}

Reallocating buyers in this way eventually saturates some obedience constraint. In this case, once the mass $\delta$ of types $2$ and $3$ shifted downward reaches $0.15$, the obedience constraint in segment $1$ becomes binding: given that segment’s new composition, the seller is indifferent between prices $1$ and $2$. \cref{fig:ill_example_2b} shows the resulting segmentation, with the orange dot indicating the binding constraint. Note that the uniform price $2$ is optimal in every segment, while the segmentation remains efficient. Consequently, this segmentation maximizes consumer surplus, achieving a seller-consumer surplus combination at the bottom-right corner of \citeauthor{bbm15}'s \citeyearpar{bbm15} welfare triangle.

However, this segmentation cannot be optimal if the designer has redistributive preferences. The reason is that the designer can implement an alternative reallocation of buyers across segments that strictly improves his objective (and, more generally, weakly improves any redistributive welfare criterion). Specifically, he can exchange a mass $\varepsilon$ of type $2$ agents from segment $2$ with an equal mass of type $3$ agents from segment $1$, as illustrated in \cref{fig:ill_example_2b}. We refer to such reallocations as \emph{redistributive transfers} (\cref{def:red_trans}). Redistributive transfers leave the obedience constraints in the lower segment unaffected, but they may alter those in the higher segment. In the current example, since the obedience constraint associated with price $3$ is slack in segment $2$, this swap does not violate obedience for sufficiently small $\varepsilon$. Mathematically, redistributive transfers are always welfare-improving because redistributive welfare functions satisfy the supermodularity condition \labelcref{cond:R}.

Executing the redistributive transfer in \cref{fig:ill_example_2b} for a value of $\varepsilon$ that makes the obedience constraint in segment $2$ bind yields the segmentation shown in \cref{fig:ill_example_3a}. This segmentation maximizes consumer surplus. However, it is the only consumer-surplus-maximizing segmentation that can be optimal under strictly redistributive welfare functions.\footnote{This unique selection among consumer-surplus-maximizing segmentations is, however, not a generic property in environments with more types.} This follows from our main result, \cref{thm:opt_red_seg}, which characterizes the class of \emph{redistributively saturated} segmentations (\cref{def:saturation}) as comprising exactly those segmentations that are optimal for strictly redistributive welfare functions. This illustrates how redistributive preferences can serve as a selection device among the multiplicity of consumer-surplus-maximizing segmentations (i.e., among the segmentations leading to the bottom-right corner of \citeauthor{bbm15}'s \citeyearpar{bbm15} triangle). The resulting segmentation implements a relatively mild form of progressive pricing, which we term \emph{weak monotonicity} (\cref{def:weakly_monotone,cor:weak_mon}): whenever a buyer of some type consumes at a given price, there must \emph{exist} a buyer of higher type consuming at a (weakly) higher price. Under the present weighted-surplus objective, this segmentation is optimal if and only if $\lambda(2) \leq 4 =4 \times \lambda(3)$.

Note that the segmentation depicted in \cref{fig:ill_example_3a} still induces some type-$3$ buyers to consume at a price lower than some type-$2$ buyers. This cannot be ``fixed'' simply with the redistributive transfer described in \cref{fig:ill_example_2b} due to the binding obedience constraint in segment $2$. When $\lambda(2) > 4$, however, the designer assigns sufficient priority to type-$2$ buyers to justify another reallocation of buyers across segments: one that bypasses the binding obedience constraint in segment $2$ by supplementing the original redistributive transfer with an additional upward transfer of some mass of type-$3$ buyers from segment $2$ to a newly created segment $3$. \Cref{fig:ill_example_3a} illustrates these transfers and their associated proportional masses that exactly maintain the binding obedience constraint in segment $2$. This upward transfer of type-$3$ buyers entails destroying part of their surplus in excess of the surplus gained by type-$2$ buyers through the redistributive transfer. This naturally raises the question of when such a trade-off is welfare-improving. Our notion of \emph{strongly redistributive welfare functions} (\cref{def:strong_red}) specifically characterizes the class of welfare objectives for which these trade-offs are always beneficial. 

In the present example, applying this definition yields the simple condition $\lambda(2) > 4$. When this condition is satisfied, the designer is willing to implement the combination of transfers depicted in \cref{fig:ill_example_3a} until segment $1$ no longer contains any type-$3$ buyer, resulting in the segmentation shown in \cref{fig:ill_example_3b}. This segmentation is, once again, redistributively saturated (\cref{def:saturation}). Within the class of such segmentations, it implements an extreme form of progressive pricing, which we refer to as \emph{strong monotonicity} (\cref{def:strongly_monotone}): if a buyer of some type consumes at a given price, then \emph{every} higher type pays (weakly) higher prices.\footnote{Strong monotonicity can be interpreted as a discrete analogue of a monotone partitional information structure (see \cref{sec:red_order}).} This constitutes an instance of \cref{prop:strong_mon}, which establishes that, for any strongly redistributive welfare function, the optimal segmentation is strongly monotone.

Crucially, the segmentation depicted in \cref{fig:ill_example_3b} is \emph{not} consumer-surplus maximizing. The reduction in consumer surplus results from reallocating some high-types into higher-price segments. This reallocation enables the seller, as a by-product, to extract additional surplus from these buyers, thereby generating extra profit, which we refer to as a \emph{redistributive rent}. More generally, it is a common implication of strongly monotone segmentations that they induce redistributive rents in favor of the seller. Our \cref{thm:rent} characterizes precisely the conditions under which this occurs. The presence of redistributive rents demonstrates the trade-off between maximizing aggregate consumer surplus and channeling the gains from segmentation toward low types.

\section{Redistributive Segmentations and Welfare}

We start by proving a preliminary result that simplifies the analysis: it is without loss of generality to restrict the search for optimal redistributive segmentations to the \emph{efficient} ones. We then characterize the optimal redistributive segmentations. Finally, we discuss the economic implications of our characterization. First, we show that optimal redistributive segmentations must exhibit a \emph{monotonic} structure that induces the seller to price \emph{progressively}. Second, we show that when the redistributive concerns are \emph{strong}, optimal redistributive segmentations can lead the seller to obtain a \emph{strictly positive} share of the surplus.  

\subsection{Efficiency}\label{sec:efficiency}

A segmentation is efficient if it leads the seller to serve all the buyers. Let $\Omega = \bigl\{(\theta,p)\in\Theta \times P \ | \ \theta\geq p \bigr\}$. For any $\mu\in\Delta(\Theta)$, the set of efficient segmentations of $\mu$ is defined by
\begin{equation*}
    \Sigma^{\star}(\mu) = \bigl\{\sigma\in\Sigma(\mu) \ | \ \supp(\sigma)\subseteq \Omega \bigr\}.
\end{equation*}

\begin{lemma}[Efficiency]\label{lem:efficiency}
    For any $w\in \mathcal{R}$ and $\mu\in\Delta(\Theta)$, there exists an optimal segmentation for the problem \labelcref{eqn:designer_pb} that belongs to $\Sigma^{\star}(\mu)$.
\end{lemma} 
\begin{proof}
    Let $\mu\in\Delta(\Theta)$ and consider $\sigma\in\Sigma(\mu)$ with $\supp(\sigma)\not\subset \Omega$. Therefore, there must exist $(\theta,p)\in\supp(\sigma)$ such that $\theta<p$. Consider reallocating the entire mass of buyers $\sigma(\theta,p)$ to $(\theta,\theta)$. This reallocation violates neither \labelcref{eqn:market_consistency} nor \labelcref{eqn:obedience}. The total mass of buyers of type $\theta$ is still equal to $\mu(\theta)$ so \labelcref{eqn:market_consistency} is preserved. Furthermore, this reallocation relaxes \labelcref{eqn:obedience} in both segments $\theta$ and $p$. Finally, since $w(\theta, p)= 0 \leq w(\theta,\theta)$, the aggregate welfare weakly increases. In fact, aggregate welfare weakly increases without hurting the seller or the buyers---a (weak) Pareto improvement.
\end{proof}
\cref{lem:efficiency} establishes that segmenting a market in a redistributive way never hinders efficiency. The intuition behind that result is similar to \citeauthor{hagpanah_siegel_jpe}'s (\citeyear{hagpanah_siegel_jpe}) main result: any inefficient segmentation is Pareto-dominated (with respect to the seller's profit and consumer surplus) by some efficient one. For our class of welfare functions, such Pareto improvements always increase aggregate welfare, since $w(\theta,\, \cdot \,)$ is a decreasing function for each $\theta\in \Theta$.

Operationally, \cref{lem:efficiency} implies that we can restate the segmentation design problem \labelcref{eqn:designer_pb} as follows:
\begin{equation}\label{eqn:designer_pb_efficient}
    \max_{\sigma\in\Sigma^{\star}(\mu)} \; \sum_{(\theta,p)\in \Theta\times P} w(\theta, p) \, \sigma(\theta,p). \tag{$\mathrm{P}^{\star}_{w,\mu}$}
\end{equation}

For any $w\in\mathcal{R}$ and $\mu\in\Delta(\Theta)$, we say that $\sigma$ is \labelcref{eqn:designer_pb_efficient}-optimal if it solves the optimization program \labelcref{eqn:designer_pb_efficient}. We call \emph{optimal redistributive segmentations} all segmentations that are \labelcref{eqn:designer_pb_efficient}-optimal for some $\mu\in\Delta(\Theta)$ and $w\in\mathcal{R}$.

\subsection{Characterization of Optimal Segmentations}

The characterization of \labelcref{eqn:designer_pb_efficient}-optimal segmentations, stated below as \cref{thm:opt_red_seg}, is given in terms of the conditions described in \cref{def:saturation} (see below for a detailed explanation).
\begin{definition}[Redistributive Saturation]\label{def:saturation}
    Let $\mu\in\Delta(\Theta)$. A segmentation $\sigma\in\Sigma^{\star}(\mu)$ is \emph{redistributively saturated} (\emph{saturated} for short) if the following conditions hold:
    \begin{enumerate}[(a)]
        \item In any segment $p \in \supp(\sigma_{P})$ such that $p<\max \supp(\sigma_{P})$, the seller is indifferent between charging $p$ and some other price $q > p$;

        \item Let $p\in \supp(\sigma_{P})$ and $\theta\in\supp\bigl(\sigma(\, \cdot \, | \, p)\bigr)$. Then in every segment $p'\in\supp(\sigma_{P})$ such that $p<p'\leq \theta$ the seller is indifferent between charging $p'$ and $\theta$.
    \end{enumerate}
\end{definition}
An efficient segmentation is saturated if (a) in every segment, except maybe for the one with the highest price, the seller faces more than one optimal price, and (b) if some consumers of type $\theta$ belong to a segment with price $p$, then consumers of type $\theta$ must be sufficiently represented in higher segments so that $\theta$ is an optimal price in those segments. 

\begin{theorem}[Optimal redistributive segmentations]\label{thm:opt_red_seg}
Fix $\mu\in\Delta(\Theta)$. The following claims hold:
\begin{enumerate}[(i)]
    \item If $w\in\mathcal{R}$, then there exists a \labelcref{eqn:designer_pb_efficient}-optimal segmentation which is saturated. Furthermore, if $w\in\overline{\mathcal{R}}$, then every \labelcref{eqn:designer_pb_efficient}-optimal segmentation is saturated.
    \item Conversely, if $\sigma\in\Sigma^{\star}(\mu)$ is saturated, then there exists $w\in\mathcal{R}, w\neq 0$, such that $\sigma$ is \labelcref{eqn:designer_pb_efficient}-optimal.
\end{enumerate}
\end{theorem}

The above result establishes that the set of saturated segmentations and optimally redistributive ones are (almost) equal. Claim (i) states that an optimal segmentation must be saturated if the welfare function is strictly redistributive ($w \in \overline{\mathcal{R}}$). If $w \in \mathcal{R}$ is not strictly redistributive, then there may be indifferences that make some non-saturated segmentations also optimal, but a saturated optimal one always exists. Claim (ii) states that, conversely, any saturated segmentation is optimal for some redistributive objective. As a result, any property that is satisfied by saturated segmentations is also satisfied by optimal redistributive segmentations. We leverage this result to draw economic implications of optimal redistributive segmentations in \cref{sec:monotone,sec:rents}; \cref{sec:proofthm1} outlines its proof.\footnote{From now on, with slight abuse of terminology, we use the terms optimally redistributive segmentation and saturated segmentation interchangeably.} 

\subsubsection{Intuitive Explanation of \cref{thm:opt_red_seg}}

\Cref{thm:opt_red_seg} is best understood in light of what saturated segmentations achieve. Saturated segmentations are those for which no reallocation of buyers between segments could increase the expectation of \emph{every} redistributive welfare function without leading to a violation of the seller's obedience constraints. At this stage, two questions are in order: (1) Which type of reallocation increases the expectation of every redistributive welfare function? (2) When is it not possible to perform such reallocations without violating the seller's obedience constraints?

\Cref{prop:red_order}, formally stated in \cref{sec:proofthm1}, answers question (1) by identifying the two marginal reallocations of buyers that increase any redistributive welfare function. First, \emph{downward buyer transfers} (\cref{def:down_trans}) move buyers to segments with lower prices. Second, \emph{redistributive buyer transfers} (\cref{def:red_trans}) swap a given mass of low-type buyers in a high-price segment for an equal mass of higher-type buyers in a lower-price segment. \Cref{fig:red_down_trans} illustrates both transfers on an efficient segmentation (see also the example in \cref{fig:ill_example_2}).
\begin{figure}[h]
    \begin{center}
        \begin{tikzpicture}[scale=1, transform shape]
            \draw[step=1cm,mBlack,very thin,opacity=0.2] (0,0) grid (6,6);
            \draw[thick,-stealth] (0,0) -- (6,0) node[anchor=west,scale=1.25] {$\theta$};
            \draw[thick,-stealth] (0,0) -- (0,6) node[anchor=south,scale=1.25] {$p$};
        
            \draw[thick,mBlack,opacity=0.5] (0,0) -- (6,6);
            \fill[mBlack,fill opacity=0.05] (0,0) -- (6,6) -- (6,0) -- cycle;
            
            \placepoints{0}{1}{0}{3pt}{mBlack}{0.8};
            \placepoints{3}{5}{0}{3pt}{mBlack}{0.8};
            \placepoints{1}{2}{1}{3pt}{mBlack}{0.8};
            \placepoints{4}{6}{1}{3pt}{mBlack}{0.8};
            \placepoints{2}{4}{2}{3pt}{mBlack}{0.8};
            \placepoints{4}{6}{4}{3pt}{mBlack}{0.8};
            
            \draw[-stealth,thick,mBlue,opacity=1, very thick,opacity=0.75] (5,4) to node[midway, right, very thick,scale=1.25,opacity=0.75] {$\delta$} ++(0,-2);
            
            \draw[-stealth, very thick, mOrange,opacity=0.75] (3,2) to node[midway, left, very thick,scale=1.25,opacity=0.75] {$\varepsilon$} ++(0,-1);
            \draw[-stealth, very thick, mOrange,opacity=0.75] (4,1) to node[midway, right, very thick,scale=1.25,opacity=0.75] {$\varepsilon$} ++(0,0.89);
        \end{tikzpicture}
    \end{center}
    \caption{The shaded region represents the set $\Omega$. A downward transfer involving a mass $\delta$ of buyers (in blue) and a redistributive transfer involving a mass $\varepsilon$ of buyers (in orange).}
    \label{fig:red_down_trans}
\end{figure}

Downward transfers are desirable under redistributive welfare functions as they decrease the price charged to some buyers. Redistributive transfers are also desirable under redistributive welfare functions due to condition \labelcref{cond:R}. The main contribution of \cref{prop:red_order} is to establish the converse: any buyer reallocation that increases the expected value of all redistributive welfare functions must necessarily consist of some (finite) combination of downward and redistributive buyer transfers. This result provides a complete characterization of welfare-improving buyer reallocations under redistributive objectives.

While downward and redistributive buyer transfers are the key to making segmentations more redistributive, these transfers interfere with the seller's obedience constraints. Reallocating buyers through downward or redistributive transfers can lead the seller to deviate to a higher price in one or more segments. \cref{prop:max_order_opt,prop:charac_max_order}, also formally stated in \cref{sec:proofthm1}, answer question (2). They show that saturation is precisely the condition that ensures that no downward or redistributive transfer (or a combination thereof) can be performed without causing the seller to deviate from the recommended prices. 

To get a sense of why, let us first investigate the conditions under which a particular downward transfer is infeasible for some segmentation. Consider transferring some mass of buyers of type $\tilde{\theta}$ from a higher-price segment $p''$ to a lower-price segment $p'$, as illustrated in \cref{fig:downtransfer_ob}.
\begin{figure}[h]
    \begin{subfigure}[c]{0.495\textwidth}
            \begin{tikzpicture}[scale=0.6, transform shape]
            \node[below,scale=1.75] at (6,0) {$\tilde{\theta}$};
            \node[scale=0.5] at (6,0) {$|$};
                       
            \node[left,scale=1.75] at (0,2) {$p'$};
            \node[scale=0.5,rotate=90] at (0,2) {$|$};
            \node[left,scale=1.75] at (0,5) {$p''$};
            \node[scale=0.5,rotate=90] at (0,5) {$|$};
            \draw[step=1cm,mBlack,very thin,opacity=0.2] (0,0) grid (9,9);
            \draw[thick,- stealth] (0,0) -- (9,0) node[anchor=west,scale=1.75] {$\theta$};
            \draw[thick,- stealth] (0,0) -- (0,9) node[anchor=south,scale=1.75] {$p$};
            
            \draw[thick,mBlack,opacity=0.5] (0,0) -- (9,9);
            \fill[mBlack,fill opacity=0.05] (0,0) -- (9,9) -- (9,0) -- cycle;
            
            \node[mOrange,opacity=0.75,scale=2] at (2,2) {$\bullet$};
            \node[mOrange,opacity=0.75,scale=2] at (5,5) {$\bullet$};
            
            \coordinate (p1) at (6,5);
            \coordinate (p2) at (6,2);
            \coordinate (p3) at (8,2);
            \coordinate (p4) at (8,5);
            
            
            \node[mBlack,opacity=0.75,scale=2] at (3,2) {$\bullet$};
            \node[mBlack,opacity=0.75,scale=2] at (6,2) {$\bullet$};
            \node[mBlack,opacity=0.75,scale=2] at (6,5) {$\bullet$};
            \node[mOrange,opacity=0.75,scale=2] at (5,5) {$\bullet$};
            \node[mBlack,opacity=0.75,scale=2] at (7,5) {$\bullet$};
            \node[mBlack,opacity=0.75,scale=2] at (9,5) {$\bullet$};
        
            \node[mBlack,opacity=0.75,scale=2] at (4,2) {$\bullet$};
            \node[mBlack,opacity=0.75,scale=2] at (5,2) {$\bullet$};
            \node[mBlack,opacity=0.75,scale=2] at (7,2) {$\bullet$};
            \node[mBlack,opacity=0.75,scale=2] at (8,2) {$\bullet$};
            \node[mBlack,opacity=0.75,scale=2] at (9,2) {$\bullet$};
        
            \node[mBlack,opacity=0.75,scale=2] at (8,5) {$\bullet$};

            \draw[mOrange, thick, opacity=0.6, rounded corners] (2.5,1.5) rectangle (6.5,2.5);
            \draw[mOrange, thick, opacity=0.6, rounded corners] (6.5,4.5) rectangle (9.5,5.5);
            
            
            \draw[- stealth,ultra thick,mBlue,opacity=0.6] (p1)  -- ++(0,-2.8);
            \node[right, mBlue, opacity=0.6, scale=2] at (6,3.5) {$\delta$};
        
            \end{tikzpicture}
    \caption{Downward transfer.}
    \label{fig:downtransfer_ob}
    \end{subfigure}
    \begin{subfigure}[c]{0.495\textwidth}
           \centering
            \begin{tikzpicture}[scale=0.6, transform shape]
            \node[below,scale=1.75] at (5,0) {$\theta'$};
            \node[scale=0.5] at (5,0) {$|$};
            \node[below,scale=1.75] at (8,0) {$\theta''$};
            \node[scale=0.5] at (8,0) {$|$};
            \node[left,scale=1.75] at (0,2) {$p'$};
            \node[scale=0.5,rotate=90] at (0,2) {$|$};
            \node[left,scale=1.75] at (0,4) {$p''$};
            \node[scale=0.5,rotate=90] at (0,4) {$|$};
            \draw[step=1cm,mBlack,very thin,opacity=0.2] (0,0) grid (9,9);
            \draw[thick,- stealth] (0,0) -- (9,0) node[anchor=west,scale=1.75] {$\theta$};
            \draw[thick,- stealth] (0,0) -- (0,9) node[anchor=south,scale=1.75] {$p$};
            
            \draw[thick,mBlack,opacity=0.5] (0,0) -- (9,9);
            \fill[mBlack,fill opacity=0.05] (0,0) -- (9,9) -- (9,0) -- cycle;
            
            \node[mOrange,opacity=0.75,scale=2] at (2,2) {$\bullet$};
            \node[mOrange,opacity=0.75,scale=2] at (4,4) {$\bullet$};
            
            \coordinate (p1) at (5,4);
            \coordinate (p2) at (5,2);
            \coordinate (p3) at (8,2);
            \coordinate (p4) at (8,4);
            
            \node[mBlack,opacity=0.75,scale=2] at (6,4) {$\bullet$};
            \node[mBlack,opacity=0.75,scale=2] at (6,2) {$\bullet$};
            \node[mBlack,opacity=0.75,scale=2] at (3,2) {$\bullet$};
            \node[mBlack,opacity=0.75,scale=2] at (5,4) {$\bullet$};
            \node[mBlack,opacity=0.75,scale=2] at (7,4) {$\bullet$};
            \node[mBlack,opacity=0.75,scale=2] at (9,4) {$\bullet$};
        
            \node[mBlack,opacity=0.75,scale=2] at (4,2) {$\bullet$};
            \node[mBlack,opacity=0.75,scale=2] at (5,2) {$\bullet$};
            \node[mBlack,opacity=0.75,scale=2] at (7,2) {$\bullet$};
            \node[mBlack,opacity=0.75,scale=2] at (8,2) {$\bullet$};
            \node[mBlack,opacity=0.75,scale=2] at (9,2) {$\bullet$};
        
            \node[mBlack,opacity=0.75,scale=2] at (8,4) {$\bullet$};

            \draw[mOrange, thick, opacity=0.6, rounded corners] (5.5,3.5) rectangle (8.5,4.5);
            
            
            \draw[- stealth,ultra thick,mBlue,opacity=0.6] (p1) -- ++(0,-1.85) ;
            \node[left, mBlue, opacity=0.5, scale=2] at (5,3) {$\varepsilon$};
            \draw[- stealth,ultra thick,mBlue,opacity=0.6] (p3) -- ++(0,+1.85);
            \node[right, mBlue, opacity=0.5, scale=2] at (8,3) {$\varepsilon$};
            
            \end{tikzpicture}
            \caption{Redistributive transfer.}
            \label{fig:red_transfer_ob}
    \end{subfigure}
    \caption{Points within the orange boxes represent $(\theta,p)$ for which \labelcref{eqn:obedience} tightens after a transfer.}
\end{figure}
This transfer affects the seller's obedience constraints in both segments. In the lower segment $p'$, adding these buyers increases demand at prices between $p'$ and $\tilde{\theta}$, tightening constraints for these prices (shown in the orange box in the lower segment). Simultaneously, in the higher segment $p''$, removing buyers decreases demand at prices between $p''$ and $\tilde{\theta}$, tightening constraints for prices above $\tilde{\theta}$ (represented in the orange box in the upper segment). If at least one of the obedience constraints is already binding at those prices in either segment, even an arbitrarily small transfer will violate the seller's obedience constraints, making the downward transfer infeasible.

The reasoning is similar for redistributive transfers. Consider swapping some mass of buyers between segments as shown in \cref{fig:red_transfer_ob}: buyers of type $\theta'$ move from segment $p''$ to $p'$, while buyers of type $\theta''$ move from $p'$ to $p''$. This transfer does not tighten any of the obedience constraints in segment $p'$. Indeed, it decreases the demand at all prices between $\theta'$ and $\theta''$, and does not affect the demand at any other price. However, in the higher segment $p''$, it increases demand at all these intermediate prices (highlighted by the orange box in \cref{fig:red_transfer_ob}), potentially violating obedience constraints in that segment.

Saturated segmentations are precisely those for which no downward or redistributive transfers remain feasible. Condition (ii) ensures that no buyer of type $\theta$ belonging to segment $p$ can be moved upwards, as the obedience constraint for price $\theta$ binds in every upper segment they could be moved to.\footnote{Since we are restricting attention to efficient segmentations, this corresponds to segments $p'$ such that $p<p'\leq\theta$.} This blocks all redistributive transfers, which necessarily involve moving some buyers upward. Condition (ii) also blocks some downward transfers: one cannot transfer some type $\theta$ downwards to a segment that includes any higher type $\theta'>\theta$, since this would lead to a violation of obedience in the origin segment. Transferring $\theta$ downwards to a segment without any higher types, however, could in principle be possible. Condition (i) blocks those remaining downward transfers by ensuring that even if such transfer does not lead to a violation of obedience in the origin segment, it does in the destination one.

\subsubsection{Proof of \cref{thm:opt_red_seg}}\label{sec:proofthm1}

The proof of \cref{thm:opt_red_seg} is based on a key result of independent interest. We articulate a sense in which one segmentation is \textit{more redistributive} than another by introducing a partial order on the space of efficient segmentations (\cref{def:red_order}). Under this order, a segmentation is more redistributive than another if the former generates higher aggregate welfare than the latter \emph{for every} redistributive welfare function. \Cref{prop:red_order} establishes that this occurs precisely when the more redistributive segmentation can be obtained through a sequence of downward and redistributive buyer transfers applied to the less redistributive one.\footnote{This provides an answer to a practical question: Say one observes some market segmentation. How can it be modified to make it \textit{more redistributive}?} We then demonstrate that optimal redistributive segmentations are exactly those for which no segmentation that is more redistributive exists (\cref{prop:max_order_opt}) and that this condition is equivalent to being saturated (\cref{prop:charac_max_order}).


\begin{proof}[Proof of \cref{thm:opt_red_seg}]
    We first introduce a \emph{redistributive order} on the space of efficient segmentations:\footnote{See \cref{sec:red_order} for the detailed connection between our redistributive order and existing stochastic orders, in particular, the supermodular one \citep{shaked2007,Muller2013,Meyer2015}.}
    \begin{definition}[Redistributive order]\label{def:red_order}
        Fix $\mu\in\Delta(\Theta)$, and let $\sigma, \sigma' \in\Sigma^{\star}(\mu)$. We say that $\sigma$ is \emph{more redistributive} than $\sigma'$, denoted $\sigma\succsim_{\mathcal{R}} \sigma'$, if
        \begin{equation*}
            \forall w\in \mathcal{R}, \quad \sum_{(\theta,p)\in \Theta\times P} w(\theta, p) \, \sigma(\theta,p) \geq \sum_{(\theta,p)\in \Theta\times P} w(\theta, p) \, \sigma'(\theta,p).
        \end{equation*}
        For any $\mu\in\Delta(\Theta)$, we call the binary relation $\succsim_{\mathcal{R}}$ over $\Sigma^{\star}(\mu)$ the \emph{redistributive order} and denote as $\succ_{\mathcal{R}}$ its asymmetric part.\footnote{One can show that $\succsim_{\mathcal{R}}$ is a \emph{partial order}, that is, a transitive, reflexive and antisymmetric binary relation. The first two conditions are easy to verify from the definition, the fact that it is antisymmetric follows from \cref{prop:red_order}.} We say that a segmentation $\sigma \in \Sigma^{\star}(\mu)$ is $\succsim_{\mathcal{R}}$-\textit{maximal} if there exists no other segmentation $\sigma' \in \Sigma^{\star}(\mu)$ such that $\sigma' \succ_{\mathcal{R}} \sigma$.
    \end{definition}
    
    Note that $\succsim_{\mathcal{R}}$ is not a complete order. A given segmentation may yield higher aggregate welfare than another for certain redistributive welfare functions, while the latter may yield higher aggregate welfare than the former for different redistributive welfare functions. Therefore, there might exist multiple $\succsim_{\mathcal{R}}$-maximal segmentations of a given market.
    
    We call a \emph{buyer mass transfer} (in short, a \emph{transfer}) any $t\colon \Omega\to \mathbb{R}$ that satisfies $\sum_{p\leq \theta} t(\theta,p) = 0$ for any $\theta\in \Theta$. Intuitively, transfers perturb a segmentation $\sigma$ by moving buyers between segments, with $t(\theta,p)$ denoting the mass of buyers of type $\theta$ that is added to or removed from a segment with price $p$. We now define two particular classes of transfers.
    \begin{definition}[Downward transfers]\label{def:down_trans}
        A transfer $t$ is \emph{downward} if there exist $\delta\geq 0$ and $\tilde{\theta}, p', p'' \in \Theta$ such that $p'<p''\leq \tilde{\theta}$, and
        \begin{equation*}
            t(\theta,p)= \left\{\begin{array}{ll}
                \delta  & \text{if $(\theta,p) =(\tilde{\theta},p')$} \\
                -\delta  &  \text{if $(\theta,p) = (\tilde{\theta},p'')$} \\
                0  & \text{otherwise}
        \end{array}\right.
        \end{equation*}
    \end{definition}
    \begin{definition}[Redistributive transfers]\label{def:red_trans}
        A transfer $t$ is \emph{redistributive} if there exists $\varepsilon \geq 0$ and $(\theta', p''), (\theta'', p') \in \Omega$ such that $\theta'<\theta''$, $p'<p''$, and
        \begin{equation*}
        t(\theta,p)=\left\{
        \begin{array}{ll}
            \varepsilon  & \text{if $(\theta,p) \in \bigl\{(\theta',p'), (\theta'',p'') \bigr\}$} \\[5pt]
           -\varepsilon  &  \text{if $(\theta,p)\in \bigl\{(\theta',p''), (\theta'',p' )\bigr\}$} \\[5pt]
           0  & \text{otherwise}
        \end{array}\right.
        \end{equation*}
    \end{definition}
    
    \Cref{prop:red_order} shows that performing a sequence of such transfers is not only \emph{sufficient} to make a segmentation more redistributive (according to \cref{def:red_order}) but also \emph{necessary}.
    Let $\mathcal{T}$ be the set of all downward and redistributive transfers.
    \begin{proposition}\label{prop:red_order}
        Fix $\mu \in \Delta(\Theta)$, and let $\sigma, \sigma' \in \Sigma^{\star}(\mu)$. The two following claims are equivalent:
        \begin{enumerate}[(i)]
            \item $\sigma \succsim_{\mathcal{R}} \sigma'$;
            \item There exists $n\in\mathbb{N}$, $n>0$, and $(t_{i})_{i\in\{1,\dots,n\}}\in \mathcal{T}^{n}$ such that
        \begin{equation*}
            \sigma = \sigma' + \sum_{i=1}^{n} t_{i}.
        \end{equation*}
        \end{enumerate}
    \end{proposition}
    
    The proof of \cref{prop:red_order}, which can be found in \cref{app:proof_red_order}, relies on duality arguments akin to those employed by \cite{Muller2013} and \cite{Meyer2015}. Specifically, we show that the space of redistributive welfare functions (up to an additive constant) is dual to the cone of downward and redistributive transfers.
    
    The definition of the redistributive order implies that, given a market $\mu\in\Delta(\Theta)$ and a redistributive welfare function $w\in\mathcal{R}$, the set of \labelcref{eqn:designer_pb_efficient}-optimal segmentations must include some $\succsim_{\mathcal{R}}$-maximal segmentation. \cref{prop:max_order_opt}, the proof of which is in \cref{app:proof_satur}, establishes the converse: for any aggregate market $\mu\in\Delta(\Theta)$, any $\succsim_{\mathcal{R}}$-maximal segmentation of $\mu$ can be rationalized as optimal for some redistributive welfare function. Its proof relies on the separating hyperplane theorem.
    \begin{lemma}\label{prop:max_order_opt} 
    Fix $\mu\in\Delta(\Theta)$. The following claims hold:
        \begin{enumerate}[(i)]
            \item For any $w \in \mathcal{R}$, there exists a \labelcref{eqn:designer_pb_efficient}-optimal segmentation of $\mu$ that is $\succsim_{\mathcal{R}}$-maximal. Furthermore, if $w \in \overline{\mathcal{R}}$, then any \labelcref{eqn:designer_pb_efficient}-optimal segmentation of $\mu$ is $\succsim_{\mathcal{R}}$-maximal.
            \item Conversely, if $\sigma\in\Sigma^{\star}(\mu)$ is $\succsim_{\mathcal{R}}$-maximal, then there exists $w \in \mathcal{R}, w\neq 0$, such that $\sigma$ is \labelcref{eqn:designer_pb_efficient}-optimal.
        \end{enumerate}
    \end{lemma}
    
    \cref{prop:max_order_opt} establishes the (almost) equality between $\succsim_{\mathcal{R}}$-maximal segmentations and optimal redistributive ones. Consequently, any property satisfied by the former must also be satisfied by the latter. 
    
    \cref{prop:charac_max_order}, the proof of which is in \cref{app:proof_satur_2}, concludes the proof of \cref{thm:opt_red_seg} by showing that saturation is necessary and sufficient for  $\succsim_{\mathcal{R}}$-maximality. Its proof consists of first showing that saturation characterizes the segmentations for which no \textit{single} downward or redistributive transfer can be feasibly performed, and then establishing a local-to-global result showing that this is also sufficient for no \textit{sequence} of downward and redistributive transfers to be feasible.
    \begin{lemma}\label{prop:charac_max_order}
        Fix $\mu\in\Delta(\Theta)$. Then, $\sigma\in\Sigma^{\star}(\mu)$ is $\succsim_{\mathcal{R}}$-maximal if and only if it is saturated.
    \end{lemma}
\end{proof}


\subsection{Monotonicity and Progressive Pricing}\label{sec:monotone}

In this subsection, we show an important economic implication of \cref{thm:opt_red_seg}. Optimal redistributive segmentations exhibit a \emph{positive assortative} structure. They assign low-type buyers to low-price segments and high-type buyers to high-price segments. Furthermore, this positive-assortative structure is strengthened as redistributive concerns grow. 

\subsubsection{Weak Monotonicity}

We start by introducing a weak notion of positive assortativity for segmentations:
\begin{definition}[Weakly monotone segmentations]
    \label{def:weakly_monotone}
    For any $\mu \in \Delta(\Theta)$, a segmentation $\sigma\in\Sigma^{\star}(\mu)$ is \emph{weakly monotone} if for any $p,p'\in \supp(\sigma_{P})$ such that $p<p'$,
    \begin{equation*}
        \max \supp\bigl(\sigma(\, \cdot \, | \, p)\bigr) \leq \max  \supp\bigl(\sigma(\, \cdot \, | \, p')\bigr). 
    \end{equation*}
\end{definition}
\begin{figure}[t]
\begin{subfigure}[c]{0.48\textwidth}
    \begin{tikzpicture}[scale=0.6, transform shape]
    \draw[step=1cm,mBlack,very thin,opacity=0.2] (0,0) grid (9,9);
    \draw[thick,-stealth] (0,0) -- (9,0) node[anchor=west,scale=1.75] {$\theta$};
    \draw[thick,-stealth] (0,0) -- (0,9) node[anchor=south,scale=1.75] {$p$};
    
    \draw[thick,mBlack,opacity=0.5] (0,0) -- (9,9);
    \fill[mBlack,fill opacity=0.05] (0,0) -- (9,9) -- (9,0) -- cycle;

    \placepoints{0}{4}{0}{4.25pt}{mBlack}{0.8};
    \placepoints{2}{6}{2}{4.25pt}{mBlack}{0.8};
    \placepoints{8}{8}{2}{4.25pt}{mBlack}{0.8};
    \placepoints{5}{8}{5}{4.25pt}{mBlack}{0.8};
    \placepoints{7}{9}{7}{4.25pt}{mBlack}{0.8};

    \end{tikzpicture}
    \caption{Weak monotonicity.}
    \label{fig:monotone}
\end{subfigure}
\begin{subfigure}[c]{0.48\textwidth}
    \begin{tikzpicture}[scale=0.6, transform shape]
    \draw[step=1cm,mBlack,very thin,opacity=0.2] (0,0) grid (9,9);
    \draw[thick,-stealth] (0,0) -- (9,0) node[anchor=west,scale=1.8] {$\theta$};
    \draw[thick,-stealth] (0,0) -- (0,9) node[anchor=south,scale=1.8] {$p$};
    
    \draw[thick,mBlack,opacity=0.5] (0,0) -- (9,9);
    \fill[mBlack,fill opacity=0.05] (0,0) -- (9,9) -- (9,0) -- cycle;

    \placepoints{0}{3}{0}{4.25pt}{mBlack}{0.8};
    \placepoints{3}{6}{3}{4.25pt}{mBlack}{0.8};
    \placepoints{7}{9}{7}{4.25pt}{mBlack}{0.8};
    \placepoints{9}{9}{9}{4.25pt}{mBlack}{0.8};


    \end{tikzpicture}
    \caption{Strong monotonicity.}
    \label{fig:strongly_monotone}
\end{subfigure}
\caption{Monotone segmentations.}
\label{fig:optimal_segmentations}
\end{figure}

A segmentation is weakly monotone if, for any pair of segments, the highest type assigned to the lower-price segment is not strictly greater than the highest type assigned to the higher-price segment. \Cref{fig:monotone} illustrates weakly monotone segmentations.

A direct implication of \cref{thm:opt_red_seg} is that all optimal redistributive segmentations are weakly monotone. Indeed, by condition (b) in \cref{def:saturation}, for any saturated segmentation, if a type $\theta$ is in a segment $p \in \supp(\sigma_{P})$, then it must belong to any other segment $p' \in \supp(\sigma_{P})$ such that $p < p' \leq \theta$, which implies weak monotonicity.\footnote{We state the result for strictly redistributive welfare functions to clarify the exposition. If $w\in \mathcal{R}$, then for any $\mu\in\Delta(\Theta)$, there exists a \labelcref{eqn:designer_pb_efficient}-optimal and weakly monotone segmentation of $\mu$.} 

\begin{corollary} \label{cor:weak_mon}
    For any $\mu \in \Delta(\Theta)$ and $w \in \overline{\mathcal{R}}$, every \labelcref{eqn:designer_pb_efficient}-optimal segmentation is \emph{weakly monotone}.
\end{corollary}

\Cref{cor:weak_mon} establishes that redistributive segmentations induce a form of \emph{progressive pricing}, whereby high-types are charged higher prices than low-types. In weakly monotone segmentations, for every buyer of a given type, there exists a buyer of a higher type that is charged a (weakly) higher price. This notion of progressivity is mild, since some buyers of a lower type might still pay higher prices than some buyers of higher type. This outcome arises in regions where the support of different segments overlap, i.e., when buyers are \emph{randomly} assigned to segments.

\begin{figure}
    \centering
    \begin{subfigure}[b]{0.4955\linewidth}
        \begin{center}
            \begin{tikzpicture}[scale=0.8, transform shape]
            \draw[step=1cm,mBlack,very thin,opacity=0.2] (0,0) grid (7,7);
            \draw[thick,-stealth] (0,0) -- (7,0) node[anchor=west,scale=1.25] {$\theta$};
            \draw[thick,-stealth] (0,0) -- (0,7) node[anchor=south,scale=1.25] {$p$};
        
            \node[left,mBlack,scale=1.25] at (0,1) {$p$};
            \node[left,mBlack,scale=1.25] at (0,3) {$\theta_{k}$};
    
            \node[left,mBlack,scale=1.25] at (0,6) {$\theta_{\ell}$};
        
            \node[below,mBlack,scale=1.25] at (3,0) {$\theta_{k}$};
            \node[below,mBlack,scale=1.25] at (4,0) {$\theta_{k+1}$};
            \node[below,mBlack,scale=1.25] at (6,0) {$\theta_{\ell}$};

            \draw[thick,mBlack,opacity=0.5] (0,0) -- (7,7);
            \fill[mBlack,fill opacity=0.05] (0,0) -- (7,7) -- (7,0) -- cycle;
        
            \node[mOrange,scale=1.5] at (3,3) {$\bullet$};
            \node[mOrange, scale=1.5] at (4,3) {$\bullet$};
            \node[mOrange,scale=1.5] at (1,1) {$\bullet$};
            \node[mOrange,scale=1.5] at (2,1) {$\bullet$};
            \node[mBlack,scale=1.5] at (3,1) {$\bullet$};
            \node[mBlack,scale=1.5] at (4,1) {$\bullet$};
            \node[mBlack,scale=1.5] at (5,3) {$\bullet$};
            \node[mBlack,scale=1.5] at (6,3) {$\bullet$};
        
            \end{tikzpicture}
        \end{center}
        \caption{A saturated weakly monotone segmentation.}
        \label{fig:strongly_red_left}
    \end{subfigure}
    \begin{subfigure}[b]{0.4955\linewidth}
        \begin{center}
            \begin{tikzpicture}[scale=0.8, transform shape]
            \draw[step=1cm,mBlack,very thin,opacity=0.2] (0,0) grid (7,7);
            \draw[thick,-stealth] (0,0) -- (7,0) node[anchor=west,scale=1.25] {$\theta$};
            \draw[thick,-stealth] (0,0) -- (0,7) node[anchor=south,scale=1.25] {$p$};
        
            \node[left,mBlack,scale=1.25] at (0,1) {$p$};
            \node[left,mBlack,scale=1.25] at (0,3) {$\theta_{k}$};
    
            \node[left,mBlack,scale=1.25] at (0,6) {$\theta_{\ell}$};
        
            \node[below,mBlack,scale=1.25] at (3,0) {$\theta_{k}$};
            \node[below,mBlack,scale=1.25] at (4,0) {$\theta_{k+1}$};
            \node[below,mBlack,scale=1.25] at (6,0) {$\theta_{\ell}$};

            \draw[thick,mBlack,opacity=0.5] (0,0) -- (7,7);
            \fill[mBlack,fill opacity=0.05] (0,0) -- (7,7) -- (7,0) -- cycle;
        
            \draw[-stealth,very thick,mOrange,opacity=0.7] (3,3) to node[midway, left,scale=1.25] {$\varepsilon$} ++(0,-1.9);
            \draw[-stealth,very thick,mOrange,opacity=0.7] (4,1) to node[midway, right,scale=1.25] {$\varepsilon$} ++(0,+1.9);
            \draw[-stealth,very thick,mBlue,opacity=0.7] (6,3) to node[midway, right,scale=1.25] {$\frac{\theta_{k+1}}{\theta_{k+1}-\theta_{k}}\, \varepsilon$} ++(0,+2.9);
        
            \node[mOrange,scale=1.5] at (3,3) {$\bullet$};
            \node[mOrange, scale=1.5] at (4,3) {$\bullet$};
            \node[mOrange,scale=1.5] at (1,1) {$\bullet$};
            \node[mOrange,scale=1.5] at (2,1) {$\bullet$};
            \node[mBlack,scale=1.5] at (3,1) {$\bullet$};
            \node[mBlack,scale=1.5] at (4,1) {$\bullet$};
            \node[mBlack,scale=1.5] at (5,3) {$\bullet$};
            \node[mBlack,scale=1.5] at (6,3) {$\bullet$};
            \node[mOrange,scale=1.5] at (6,6) {$\bullet$};
        
            \end{tikzpicture}
        \end{center}
        \caption{A compensated redistributive transfer.}
        \label{fig:strongly_red_right}
    \end{subfigure}
    \caption{The left panel illustrates a saturated segmentation where some buyers $\theta_{k+1}$ consume at a lower price than some buyers $\theta_k$. Points in orange represent $(\theta,p)$ for which \labelcref{eqn:obedience} binds.}
    \label{fig:strongly_red}
\end{figure}

\subsubsection{Strengthening Monotonicity}

\paragraph{Compensated redistributive transfers}\label{sec:compensated_transfers}

Regions of random assignment are exactly where unresolved redistributive tensions remain in saturated segmentations. Consider the segmentation shown in \cref{fig:strongly_red_left}: some buyers of type $\theta_k$ are placed in the higher segment to ``make room'' for some buyers of type $\theta_{k+1}$ in the lower segment. These $\theta_k$ buyers would be better off in the lower segment, but in a saturated segmentation this cannot be achieved by simple downward or redistributive transfers without violating \labelcref{eqn:obedience}. In particular, any redistributive transfer between types $\theta_k$ and $\theta_{k+1}$ would induce the seller to deviate to price $\theta_{k+1}$ in the upper segment.

Yet such a reallocation of types $\theta_k$ to the lower segment can be feasibly implemented through an appropriate sequence of transfers. Consider the construction illustrated in \cref{fig:strongly_red_right}: first, a redistributive transfer swaps a mass $\varepsilon$ of types $\theta_k$ and $\theta_{k+1}$; then an upward transfer moves a mass $\frac{\theta_{k+1}}{\theta_{k+1}-\theta_{k}}\varepsilon$ of higher-type buyers $\theta_{\ell}$ to a segment where they are charged exactly their willingness to pay. Removing this mass of type $\theta_{\ell}$ exactly offsets the effect of adding an $\varepsilon$ mass of types $\theta_{k+1}$ to segment $\theta_k$, preserving obedience. Moreover, because the transferred $\theta_{\ell}$ buyers are the highest types in segment $\theta_k$ and are moved to $(\theta_{\ell}, \theta_{\ell})$, \labelcref{eqn:obedience} remains satisfied. Thus, whenever a segmentation contains multiple types at the overlap between segments, one can construct such a feasible combination of redistributive and upward transfers. We refer to this construction as \emph{compensated redistributive transfers}.

\paragraph{Strongly redistributive welfare functions}

Compensated redistributive transfers may be desirable under welfare functions that reflect \emph{strong} redistributive concerns. That is, welfare functions under which the welfare gain from lowering the price paid by a buyer $\theta_k$ offsets not only the welfare loss from an equivalent price increase for a buyer $\theta_{k+1}$ but also potential welfare losses borne by other higher types. Specifically, the compensated redistributive transfer illustrated in \cref{fig:strongly_red} is beneficial whenever the net welfare gain from the redistributive transfer between types $\theta_k$ and $\theta_{k+1}$ exceeds by a factor of $\frac{\theta_{k+1}}{\theta_{k+1}-\theta_{k}}$ the welfare loss from raising the price paid by a buyer of type $\theta_{\ell}$ from $\theta_k$ to $\theta_{\ell}$. 

We now define a class of welfare functions for which \emph{every} compensated redistributive transfer is desirable:

\begin{definition}[Strongly redistributive welfare functions]\label{def:strong_red}
    A welfare function $w\in \overline{\mathcal{R}}$ is \emph{strongly redistributive} if for any $p, \theta_k, \theta_\ell \in \Theta$ such that $p <\theta_{k} < \theta_\ell$, 
    \begin{multline}\label{cond:SR}
        \overbrace{\bigl[w(\theta_{k}, p) - w(\theta_{k},\theta_{k})\bigr] -\bigl[w(\theta_{k+1}, p)-w(\theta_{k+1}, \theta_{k})\bigr]}^{\text{\shortstack{Net welfare gain from a redistributive transfer between \\ types $\theta_{k}$ and $\theta_{k+1}$ in segments $p$ and $\theta_k$.}}} \\
        > \frac{\theta_{k+1}}{\theta_{k+1}-\theta_{k}}\underbrace{\bigl[w(\theta_{\ell},\theta_{k}) -  w(\theta_{\ell},\theta_{\ell})\bigr].}_{\text{\shortstack{Welfare loss from increasing the \\ price from $\theta_{k}$ to $\theta_{\ell}$ for type $\theta_{\ell}$.}}}\tag{$\mathrm{SR}$}
    \end{multline}
\end{definition}

\paragraph{Strong monotonicity}

Compensated redistributive transfers render any segmentation more positive-assortative by removing cases where a higher-type buyer pays less than a lower-type buyer.\footnote{Note that performing a compensated redistributive transfer on a saturated segmentation also increases the seller's profit, since the upward transfer increases the price being charged to $\theta_{\ell}$ while the redistributive transfer is profit-neutral. This aspect of compensated redistributive transfers is further examined in \cref{sec:rents}.} As redistributive concerns become stronger, more compensated redistributive transfers should be performed. This process reaches a limit once no further compensated redistributive transfers can be implemented, yielding what we call \emph{strongly monotone} segmentations:
\begin{definition}[Strongly monotone segmentations]
    \label{def:strongly_monotone}
    For any $\mu \in \Delta(\Theta)$, a segmentation $\sigma\in\Sigma^{\star}(\mu)$ is \emph{strongly monotone} if for any $p,p'\in \supp(\sigma_{P})$ such that $p<p'$, 
    \begin{equation*}
         \max \supp\bigl(\sigma(\, \cdot \, | \, p)\bigr) \leq p'. 
    \end{equation*}
\end{definition}

A segmentation is strongly monotone if buyers are assigned to segments \emph{deterministically}, except possibly for types lying at the boundary between two consecutive segments.\footnote{Our definition of weak monotonicity implies that segment supports are weak-set-ordered \citep{Che2021}. Strong monotonicity, by contrast, induces an ordering of segment supports stronger than the strong set order \citep{topkis}. Strongly monotone segmentations correspond to discrete version of \emph{monotone partitional} information structures, which partition the state space into convex sets \citep{Ivanov2021,Mensch2021,Kolotilin2022,Onuchic2023,Kolotilin2024}.} Strongly monotone segmentations, illustrated in \cref{fig:strongly_monotone}, thus represent a \emph{pure} form of progressive pricing: for any type, \emph{every} higher type is charged a (weakly) higher price.

\paragraph{Optimality of strong monotonicity}

Our next result shows that strongly redistributive welfare functions are the \emph{only} ones for which \emph{every} optimal segmentation is strongly monotone.

\begin{proposition}\label{prop:strong_mon}
      Any \labelcref{eqn:designer_pb_efficient}-optimal segmentation is strongly monotone for every $\mu\in\Delta(\Theta)$ if and only if $w\in \overline{\mathcal{R}}$ is strongly redistributive.
\end{proposition}

The proof of \cref{prop:strong_mon} appears in \cref{app:proof_mon}. We further show in the following lemma (the proof of which can be found in \cref{app:proof_rent}) that, for any given aggregate market, there is exactly one strongly monotone segmentation of that market that is saturated. This particular segmentation therefore represents the \emph{uniquely optimal} market segmentation for \emph{any} strongly redistributive welfare function.
\begin{lemma}\label{lem:unique_strongly_mon}
    For any $\mu \in \Delta(\Theta)$, there exists a unique strongly monotone segmentation $\sigma \in \Sigma^{\star}(\mu)$ that is saturated.
\end{lemma}

We now come back to \cref{ex:decreasing_pareto,ex:deacr_marg_surplus}. We show that the weighted consumer surplus welfare function is strongly redistributive if and only if Pareto weights decrease sufficiently fast in willingness to pay.\footnote{This shows incidentally that the set of strongly redistributive welfare functions is non-empty.} We also show that concave increasing transformations of consumer surplus are never strongly redistributive.
\setcounter{example}{\getrefnumber{ex:decreasing_pareto}-1} 
\begin{example}[Decreasing Pareto weights, cont'd]
    Let $\lambda\colon \Theta \to \mathbb{R}$ be a strictly decreasing function such that $\lambda(\theta)>0$ for all $\theta\in\Theta$. Define $w$ such that, for any $(\theta,p) \in \Omega$, $w(\theta,p) = \lambda(\theta)(\theta - p)$. By rearranging terms in the inequality \labelcref{cond:SR}, we can see that $w$ is strongly redistributive if and only if for any $\theta_k, \theta_\ell \in \Theta$ such that $2\leq k\leq K-1$ and $k < \ell \leq K$,
    \begin{equation}\label{eqn:strong_red_weights}
        \frac{\lambda(\theta_{k})-\lambda(\theta_{k+1})}{\lambda(\theta_{\ell})} \geq\frac{\theta_{k+1}}{\theta_{k+1}-\theta_{k}} \frac{\theta_{\ell}-\theta_k}{\theta_k-\theta_{k-1}}.
    \end{equation}
    
    Inequality \labelcref{eqn:strong_red_weights} imposes a condition on the rate of decrease of $\theta\mapsto\lambda(\theta)$. In this case, we can reinterpret \cref{prop:strong_mon} as follows: every optimal segmentation is strongly monotone if and only if Pareto weights decrease at a sufficiently fast rate. For instance, weights of the form $\lambda(\theta)=\theta^{-\gamma}$ or $\lambda(\theta)=\exp(-\gamma\theta)$ are strongly redistributive if $\gamma$ exceeds some threshold.
\end{example}

\begin{example}[Increasing concave transformations, cont'd]
    Let $u: \mathbb{R} \to \mathbb{R}$ be strictly increasing and strictly concave. Define $w$ such that, for any $(\theta, p) \in \Omega$,  $w(\theta, p) = u(\theta - p)$. In this case, one can show that condition \labelcref{cond:SR} can never be satisfied. In contrast to \cref{ex:decreasing_pareto}, when the redistributive concerns stem from a concave transformation of buyers' surplus, the welfare gain associated with redistributive transfers can never outweigh the welfare loss associated with burning the surplus of some high type. In this case, \cref{prop:strong_mon} implies that there always exists a market $\mu\in\Delta(\Theta)$ such that no optimal segmentation is strongly monotone.
\end{example}

\paragraph{The greedy redistributive algorithm} 

The simple structure of strongly monotone segmentations and their uniqueness pave the way for a constructive approach. We describe a procedure, called the \emph{greedy redistributive algorithm}, that outputs, for any market $\mu$, the unique strongly monotone segmentation of $\mu$ that is saturated.  

Strongly monotone segmentations can be seen as the result of lexicographic optimization: it benefits each type sequentially as much as possible, from the smallest to the largest. The greedy redistributive algorithm follows this logic and sequentially assigns types into segments. It starts by creating a lower segment with all buyers of type $\theta_1$ and then proceeds to include types $\theta_2$ in that segment. If all buyers $\theta_2$ can be included while keeping the price in the segment at $\theta_1$, the algorithm does so and proceeds to include types $\theta_3$. At some point in this process, the algorithm reaches a point such that including all buyers of a given type $\theta_k$ would violate \labelcref{eqn:obedience} in that segment. When the algorithm reaches this point, it fits as many types $\theta_k$ as possible in the segment so that \labelcref{eqn:obedience} becomes binding for some price and moves on to creating a new segment with the remaining mass of types $\theta_k$. This inclusion process is then repeated in the newly created segment starting at type $\theta_{k+1}$, and then proceeds sequentially as previously explained. The algorithm ends when every buyer has been assigned to some segment.

\subsection{Redistributive Rents}\label{sec:rents}

We now study how redistributive segmentations distribute surplus between the seller and the buyers. We show that if the welfare function is strongly redistributive, then optimal segmentation might allocate some of the created surplus to the seller. Hence, although the welfare function does not take into consideration the seller's profit, redistributive segmentations may not maximize total consumer surplus.

For any $\mu\in\Delta(\Theta)$, we say that a segmentation $\sigma\in\Sigma^{\star}(\mu)$ gives a \emph{rent} to the seller if the profit under $\sigma$ is strictly higher than the profit under uniform pricing: that is, if
\begin{equation*}
    \sum_{p \in P} p \sum_{\theta\geq p} \sigma(\theta,p) > p^{\star}_{\mu}\sum_{\theta\geq p^{\star}_{\mu}} \mu(\theta).
\end{equation*}

The following joint probability distribution over types and prices will prove to be of particular interest. For any market $\mu \in \Delta(\Theta)$, we define $\sigma^{\star}_{\mu} \in \Delta(\Theta \times P)$ as follows:
\begin{equation*}
    \sigma^{\star}_{\mu}(\theta, p)=\left\{
    \begin{array}{ll}
        \mu(\theta) & \text{if $\theta < p^{\star}_{\mu}$ and $p = \theta_{1}$} \\[15pt]
        \min\left\{\mu(p^{\star}_{\mu}),\displaystyle\frac{\theta_{1}}{p^{\star}_{\mu} - \theta_{1}}\displaystyle\sum_{\theta< p^{\star}_{\mu}} \mu(\theta)\right\} & \text{if $(\theta,p) = (p^{\star}_{\mu}, \theta_{1})$} \\[15pt]
        \mu(p^{\star}_{\mu})-\sigma^{\star}_{\mu}(p^{\star}_{\mu}, \theta_1) & \text{if $(\theta,p) = (p^{\star}_{\mu}, p^{\star}_{\mu})$} \\[15pt]
        \mu(\theta) & \text{if $\theta > p^{\star}_{\mu} \text{ and } p = p^{\star}_{\mu}$} \\[15pt]
        0  & \text{otherwise.}
    \end{array}\right.
\end{equation*}

\begin{figure}[t]
 \centering
        \begin{tikzpicture}[scale=1, transform shape]
            \draw[step=1cm,mBlack,very thin,opacity=0.2] (0,0) grid (6,6);
            \draw[thick,-stealth] (0,0) -- (6,0) node[anchor=west] {$\theta$};
            \draw[thick,-stealth] (0,0) -- (0,6) node[anchor=south] {$p$};
            
            \node[below] at (3,-0.1) {$p^{\star}_{\mu}$};
            
            \node[left] at (0,3) {$p^{\star}_{\mu}$};
            \node[left] at (0,0) {$\theta_1$};
            

            \node[mOrange,opacity=0.75,scale=1.5] at (0,0) {$\bullet$};
            \node[mOrange,opacity=0.75,scale=1.5] at (3,0) {$\bullet$};
            \node[mOrange,opacity=0.75,scale=1.5] at (3,3) {$\bullet$};

            \node[mBlack,opacity=0.75,scale=1.5] at (1,0) {$\bullet$};
            \node[mBlack,opacity=0.75,scale=1.5] at (2,0) {$\bullet$};
            \node[mBlack,opacity=0.75,scale=1.5] at (4,3) {$\bullet$};
            \node[mBlack,opacity=0.75,scale=1.5] at (5,3) {$\bullet$};
            \node[mBlack,opacity=0.75,scale=1.5] at (6,3) {$\bullet$};
            
            \draw[thick,mBlack,opacity=0.5] (0,0) -- (6,6);
            \fill[mBlack,fill opacity=0.05] (0,0) -- (6,6) -- (6,6) -- (6,0) -- cycle;
        \end{tikzpicture}
        \caption{The segmentation $\sigma^{\star}_{\mu}$ for a uniform price $p^{\star}_{\mu}$.}
        \label{fig:sigmabar}
\end{figure}
By construction, $\sigma^{\star}_{\mu}$ satisfies \labelcref{eqn:market_consistency} and $\supp(\sigma^{\star}_{\mu})\subset \Omega$. Therefore, if $\sigma^{\star}_{\mu}$ satisfies \labelcref{eqn:obedience} then it is a segmentation belonging to $\Sigma^{\star}(\mu)$. In this case, we refer to $\sigma^{\star}_{\mu}$ as being \emph{feasible}. Note that, if $\sigma^{\star}_{\mu}$ is feasible, then it is the unique saturated and strongly monotone segmentation of\footnote{Note that $\sigma^{\star}_{\mu}$ is constructed such that, if it is feasible, the seller is indifferent between charging $\theta_1$ and $p^{\star}_{\mu}$ in segment $\theta_1$. As such, $\sigma^{\star}_{\mu}$ satisfies conditions (a) and (b) of \cref{def:saturation} if it is feasible, implying that it is saturated. Intuitively, the feasibility of $\sigma^\star_\mu$ is therefore demanding as it requires a large mass of type $p^\star_\mu$ such that $p^\star_\mu$ is an optimal price in both segments.} $\mu$. Furthermore, $\sigma^{\star}_{\mu}$ has a very simple structure. It assigns all types lower than $p^{\star}_{\mu}$ to a ``discount'' segment where the lowest possible price $\theta_1$ is charged, and all the remaining types to a residual segment where the uniform price $p^{\star}_{\mu}$ is charged. \cref{fig:sigmabar} illustrates $\sigma^\star_\mu$.
\begin{proposition}[Redistributive Rents]\label{thm:rent}
For any $\mu \in \Delta(\Theta)$ and any strongly redistributive welfare function $w$, the two following statements are true:
    \begin{enumerate}[(i)]
        \item If $\sigma^{\star}_{\mu}$ is feasible, then it is the unique \labelcref{eqn:designer_pb_efficient}-optimal segmentation and it does not give any rent to the seller.
        \item If $\sigma^{\star}_{\mu}$ is not feasible, then the unique \labelcref{eqn:designer_pb_efficient}-optimal segmentation gives a rent to the seller.
    \end{enumerate}
\end{proposition}
\begin{proof}
A segmentation does not give any rent to the seller if and only if $p^{\star}_{\mu}$ is an optimal price in every market segment.\footnote{This follows easily from the fact that the seller's total profit is the sum of its profits on each segment.} Therefore, a strongly monotone segmentation does not give a rent to the seller only if it satisfies two conditions. First, it must generate only two market segments. One containing every type weakly lower than $p^{\star}_{\mu}$, and one containing every type weakly greater than $p^{\star}_{\mu}$. Second, \labelcref{eqn:obedience} must bind at $p^{\star}_{\mu}$ in the segment $\theta_1$. It can easily be seen that $\sigma^{\star}_{\mu}$ is the only strongly monotone segmentation that satisfies those two requirements. Furthermore, following \cref{prop:strong_mon}, for any strongly redistributive $w$, the \labelcref{eqn:designer_pb_efficient}-optimal segmentation is the unique strongly monotone one that is saturated. Therefore, if $\sigma^{\star}_{\mu}$ is feasible, then it is uniquely \labelcref{eqn:designer_pb_efficient}-optimal and does not give a rent to the seller. On the contrary, if $\sigma^{\star}_{\mu}$ is not feasible, then the \labelcref{eqn:designer_pb_efficient}-optimal segmentation must give a rent to the seller. 
\end{proof}
\Cref{thm:rent} highlights the key tradeoff generated by redistributive concerns. Benefiting poor consumers may require imposing a sacrifice on richer consumers' surplus. Strongly redistributive welfare functions precisely characterize when the welfare gains from redistribution more than compensate for the surplus loss of higher types. Given that redistributive segmentations are efficient (\cref{lem:efficiency}), as a byproduct of this loss, the seller benefits from a redistributive rent. This result is easily explainable in light of the structure of strongly redistributive segmentations (\cref{prop:strong_mon}). Such segmentations pool as many low-types as possible in low-price segments, allowing them to benefit from lower prices. Consequently, high types are separated from low types, enabling the seller to better price discriminate. This increases the seller's profit, except in the special case where $\sigma^{\star}_{\mu}$ is feasible.

\subsection{Discussion}\label{sec:discussion}

\subsubsection{Redistribution and Consumer-Surplus Maximization}\label{sec:CS_max}

The results in this paper highlight a variety of ways in which redistributive segmentations can differ from consumer surplus-maximizing segmentations \citep{bbm15}. Notably, redistributive segmentations might differ from consumer surplus-maximizing (henceforth CS-maximizing) ones in how they assign consumers to segments, in the prices they generate, and in the way they distribute surplus across the buyers and the seller.

Redistributive segmentations exhibit a positive-assortative structure, assigning consumers with a lower willingness to pay to segments with low prices and consumers with a higher willingness to pay to segments with higher prices. As seen in \cref{sec:monotone}, all redistributive segmentations satisfy weak monotonicity: if a buyer of type $\theta$ consumes at price $p$, there must exist a buyer of higher type who consumes at a higher price. When the redistributive motive is sufficiently strong, the optimal redistributive segmentations also satisfy strong monotonicity: if a buyer of type $\theta$ consumes at price $p$, every buyer of a higher type must be consuming at a higher price. CS-maximizing segmentations need not satisfy even this weak notion of monotonicity. For instance, \citeauthor{bbm15}' (\citeyear{bbm15}) greedy algorithm may not lead to a weakly monotone segmentation (taking the equivalent direct segmentation).

The monotone structure of redistributive segmentations qualifies a sense in which redistributive segmentations are ``more separating'' than CS-maximizing ones. This greater separation is echoed in the distributions of prices induced by such segmentations: while CS-maximizing segmentations only induce segments with discounts relative to the uniform price, redistributive segmentations might also induce segments with prices higher than the uniform. As such, prices might be more dispersed under redistributive segmentations. Another symptom of this greater separation is the distribution of surplus across sides of the market. As treated in \cref{sec:rents}, while CS-maximizing segmentations allocate all of the surplus generated by the segmentation to buyers, redistributive segmentations might, as a byproduct of prioritizing the welfare of poorer consumers, grant extra surplus to the seller.

Finally, it is important to note that the set of redistributive segmentations intersects with the set of CS-maximizing ones. Indeed, as a welfare function, total consumer surplus is included in the class of welfare functions we study, despite not being our focus since condition \labelcref{cond:R}, which captures our notion of redistributive concerns, is satisfied everywhere with equality. Regardless, from \cref{thm:opt_red_seg} it follows that in the set of CS-maximizing segmentations, one can always find a saturated one.

The redistributive order characterized in \cref{prop:red_order} is also helpful if one wants to use redistribution as a criterion to select among CS-maximizing segmentations. CS-maximizing segmentations cannot be improved in the redistributive order through downward transfers, but might be improvable through redistributive transfers. Departing from some CS-maximizing segmentation, sequentially performing feasible redistributive transfers will eventually result in a saturated segmentation that is also CS-maximizing. 



\subsubsection{Link to Stochastic Orders and Optimal Transport}\label{sec:red_order}

Our characterization of the redistributive order, as presented in \cref{prop:red_order}, conceptually parallels the characterization of the supermodular stochastic order developed independently by \cite{Muller2013} and \cite{Meyer2015}. The critical distinction lies in their respective objects of comparison: the supermodular stochastic order compares \emph{couplings}---bivariate joint distributions with two prescribed marginals---whereas the redistributive order compares segmentations, which consist of joint distributions where only the first marginal (the aggregate market $\mu$) is prescribed, while the second marginal (the price distribution $\sigma_{P}$) emerges endogenously from the seller's incentive constraints.\footnote{This observation holds for any persuasion problem, where the prior over the state is given and the distribution over the receiver's action is pinned down by obedience constraints \citep[see][]{bergemann2019}. Our characterization can therefore be extended to address a recurring question in the literature: When are optimal information structures monotone?} 

Therefore, the characterization of the supermodular order over couplings differs from that of the redistributive order over segmentations in a substantive way. The supermodular order is characterized only in terms of supermodular transfers (akin to redistributive transfers in our terminology), whereas the redistributive order is characterized in terms of both supermodular transfers, which preserve both marginals, and \emph{downward} transfers, which can modify the marginal distribution of prices. This distinction arises directly from the endogenous nature of segmentations' second marginal. Since downward transfers inevitably alter the price distribution, they cannot be applied to couplings, where both marginals are prescribed.\footnote{Redistributive transfers (\cref{def:red_trans}) coincide with the mass transfers used in optimal transport to show that comonotone distributions are the extreme points of the transport polytope \citep[see][Theorem 6.14]{Carlier2022}. Future research could use similar tools to characterize the extreme points of the ``persuasion polytope.''}

Observe also that the optimization problem \labelcref{eqn:designer_pb_efficient} bears similarity to the discrete optimal assignment problem (see \citealp{Galichon2018}, Chapter 3). Again, the key difference is that rather than maximizing the expectation of a function on the set of \emph{couplings}, we maximize it on the set of (efficient) segmentations.

In statistics and optimal transport theory it is known that \emph{comonotone} joint distributions maximize the expectation of supermodular functions on the set of couplings (see, e.g., Theorem 2.1 in \citealp{Puccetti2015}, and Theorem 4.3 in \citealp{Galichon2018}). We show that, similarly as for couplings, the supermodularity of redistributive welfare functions induces a positive-assortative matching of willingness to pay and prices. Furthermore, we show in \cref{prop:max_order_opt} that saturated segmentations correspond to the maximal elements of the redistributive order, which provides a parallel to how comonotone distributions correspond to the maximal elements of the supermodular stochastic order (see claim (c) of Theorem 2.1 in \citealp{Puccetti2015}).

Collectively, \cref{thm:opt_red_seg}, \cref{cor:weak_mon,prop:strong_mon} thus provide a natural analog for segmentations to the characterizations of (discrete) optimal assignments in optimal transport theory. However, since segmentations are constrained by the seller's incentives, supermodularity alone no longer guarantees comonotonicity. As discussed earlier, our concept of strongly redistributive welfare functions (\cref{def:strong_red}), which can be understood as a form of ``strong-supermodularity,'' becomes both necessary and sufficient to achieve comonotone optimal solutions.

\section{Implementation Through Price Regulation}\label{sec:implementation}

Thus far, we have analyzed the segmentation design problem from the perspective of a central planner with redistributive preferences. In practice, however, market segmentation is primarily conducted by private actors with access to consumers' data, such as platforms. This leads to the problem of implementation: How can a regulator (she) induce the seller (he) to implement a redistributive segmentation? 

If the regulator can verify the segmentation chosen by the seller, this implementation problem is trivial. The regulator can simply commit to punishing the seller if his chosen segmentation differs from her preferred one. However, it is unlikely that the regulator can perfectly observe the segmentation set by the seller. We show that observing the price distribution induced by segmentations is \emph{sufficient} for the regulator to implement her preferred segmentation.

Assume that both the regulator and the seller know the aggregate market $\mu\in\Delta(\Theta)$. The regulator communicates to the seller the segmentation $\sigma$ she wishes to implement, after which the seller chooses a market segmentation $\sigma'\in \Sigma(\mu)$. The regulator cannot observe $\sigma'$ directly but can verify and enforce punishments based on the induced price distribution\footnote{As discussed in the introduction, observing the distribution of prices is realistic; see \cref{footnote:websites}.} $\sigma'_{P}$. Consequently, for any $\mu\in\Delta(\Theta)$, the set of undetectable deviations from $\sigma$ is
\begin{equation*}
    \mathcal{D}_{\mu}(\sigma)\coloneqq \bigl\{\sigma'\in\Sigma(\mu) \ | \ \sigma^{\prime}_{P} = \sigma_{P}\bigr\}.
\end{equation*}

Since the regulator can detect any deviation that does not belong to $\mathcal{D}_{\mu}(\sigma)$, she can impose appropriate penalties and deter the seller from deviating to $\sigma'\notin\mathcal{D}_{\mu}(\sigma)$. Hence we adopt the following notion of implementability.
\begin{definition}[Price-based implementability]
    Let $\mu\in\Delta(\Theta)$. A segmentation $\sigma\in \Sigma(\mu)$ is \emph{implementable through price-based regulation} if there exists no $\sigma'\in\mathcal{D}_{\mu}(\sigma)$ such that
    \begin{equation*}
        \sum_{p\in \supp(\sigma^{\prime}_{P})} p \sum_{\theta\geq p} \sigma'(\theta,p) > \sum_{p\in \supp(\sigma_{P})} p \sum_{\theta\geq p} \sigma(\theta,p).
    \end{equation*}
\end{definition}

A segmentation can be implemented through price-based regulation if there exists no strictly profitable deviation for the seller that is undetectable by the regulator. The next proposition shows that efficiency is, in fact, sufficient for implementability.\footnote{Our notion of implementability is equivalent to the notion of credibility defined by \cite{linliu} in the context of Bayesian persuasion. Specifically, the authors define an information structure as credible if there exists no alternative information structure that would strictly increase the sender's payoff while unaffecting the marginal distribution of signal realizations. As such, \cref{prop:implementation} shows that any efficient segmentation is also seller-credible. Our notion of implementation is also equivalent to the notion of price-regulation defined by \cite{Schlom2024} in the context of second-degree price discrimination. \label{foot:credibility}}
\begin{proposition}\label{prop:implementation}
    Let $\mu\in \Delta(\Theta)$. If $\sigma\in\Sigma^{\star}(\mu)$, then $\sigma$ is implementable through price-based regulation.
\end{proposition}
\begin{proof}
Let $\mu\in\Delta(\Theta)$ and $\sigma\in\Sigma^{\star}(\mu)$. Recall that, for any $(\theta,p)\in \Omega$ we must have that $\sigma(\theta,p)=\sigma(\theta \, | \, p) \, \sigma_{P}(p)$. We thus obtain the following decomposition of the platform's profit under $\sigma$:
\begin{align*}
    \sum_{p\in \supp(\sigma_{P})} p \sum_{\theta\geq p} \sigma(\theta,p) &= \sum_{p\in \supp(\sigma_{P})} p \, \sigma_{P}(p) \overbrace{\sum_{\theta\geq p} \sigma(\theta \, | \, p)}^{\mathclap{\text{$=1$ since $\sigma\in\Sigma^{\star}(\mu)$}}} \\[5pt]
    &= \sum_{p\in \supp(\sigma_{P})} p \, \sigma_{P}(p).
\end{align*}

By definition, $\sum_{\theta \geq p} \sigma^{\prime}(\theta \, | \, p)\leq \sum_{\theta \in \Theta} \sigma^{\prime}(\theta \, | \, p) = 1$ for any $\sigma'\in\Sigma(\mu)$. Therefore, for any $\sigma'\in\mathcal{D}_{\mu}(\sigma)$, we must have in particular that
\begin{align*}
    \sum_{p\in \supp(\sigma^{\prime}_{P})} p \, \sigma^{\prime}_{P}(p) \sum_{\theta \geq p} \sigma^{\prime}(\theta \, | \, p) &\leq \sum_{p\in \supp(\sigma^{\prime}_{P})} p \, \sigma^{\prime}_{P}(p) \notag \\[5pt]
    &= \sum_{p\in \supp(\sigma_{P})} p \, \sigma_{P}(p) \, 
\end{align*}

This implies that $\sigma$ is implementable through price-based regulation.
\end{proof}

Therefore, any efficient segmentation and, following \cref{lem:efficiency}, any redistributive segmentation, can be implemented through price-based regulation. The proof follows from a simple observation: the difference in the seller's profit between two segmentations that induce the same price distribution only depends on the mass of consumers that buy the good at each price. For efficient segmentations, this mass is maximal because every buyer consumes. Therefore, there cannot be any strictly beneficial deviation.

\section{Concluding Remarks}\label{sec:conclusion}

Price discrimination has long been used as an instrument for redistribution. Many goods and services are provided at reduced prices for disadvantaged groups. For instance, students from low-income backgrounds often pay reduced tuition fees at universities, and it is common for the unemployed or the elderly to receive discounted prices at museums. The prevalence of such redistributive price discrimination indicates that personalized pricing is deemed socially acceptable as long as it is perceived to benefit those in need. This view is supported by experimental data on people's attitudes to price discrimination \citep{wu} and is expressed in policy reports \citep{WhiteHouse2015,bourreau2018regulation}. Our redistributive segmentations satisfy this social demand by ensuring that price discrimination satisfies a certain degree of price progressivity. Hence, our results inform the current policy debate over the regulation of online personalized pricing. They give qualitative guidelines regarding market segmentations that prioritize consumers with low willingness to pay.

Redistributive price discrimination is usually implemented by public and non-profit organizations \citep{legrand,steinberg}. A concern is that imposing redistributive price discrimination on the private sector could be inefficient. This paper studies how a private market with a profit-maximizing monopoly can be optimally segmented to achieve a redistributive goal. Redistributive segmentations have the twofold benefit of achieving economic efficiency (\cref{lem:efficiency}) and never harming the seller---\cref{thm:rent} shows that redistributive segmentations can even improve the seller's profit.

One might suggest that conditioning market segmentation on additional buyer information (e.g., income or group characteristics), rather than only the willingness to pay, would better target redistribution. Nevertheless, as highlighted in \cref{sec:implementation}, consumer data is usually collected and processed by private actors for which consumers' willingness to pay is the payoff-relevant variable. It might also be impossible or illegal for regulators like competition authorities to directly access income data or protected characteristics. Hence, our optimal segmentations constitute a normative benchmark against which to evaluate welfare outcomes that could be achieved by regulating data intermediaries.


\clearpage
\appendix
\begin{center}
     \textbf{\LARGE Appendices}
\end{center}

\section{A Foundation for Redistributive Welfare Functions}
\label{app:microfoundation}

\paragraph{Separating taste and income}
Consider that each buyer is characterized by a two-dimensional type $(t,y)$ drawn according to a given joint distribution. The first component, $t\in\mathbb{R}_{+}$, is the buyer's intrinsic preference (or ``pure taste'') for the good supplied by the seller. The second component, $y\in\mathbb{R}_{+}$, is the buyer's income at the time of purchase.

A buyer with type $(t,y)$ obtains utility $t+u(y-p)$ if she buys at price $p$, and utility $u(y)$ if she does not buy. The function $m\mapsto u(m)$ is non-decreasing and concave. The willingness to pay of type $(t,y)$ is defined as

\begin{equation*}
    \vartheta(t,y)
    \coloneqq
    \sup\bigl\{p\in \mathbb{R}_{+} \bigm\vert t+u(y-p)\geq u(y)\bigr\}.
\end{equation*}
We restrict attention to cases for which this supremum is finite and interior, so that
\begin{equation}
    t+u(y-\vartheta(t,y))=u(y). \tag{A.1}
\end{equation}
Thus, a buyer purchases at price $p$ if and only if $p\le \vartheta(t,y)$.

As in the main text, the planner observes the willingness to pay $\vartheta(t,y)$, but not taste $t$ and income $y$ separately. Expected consumer welfare at price $p$ is given by
\begin{equation*}
    \mathbb{E}\Bigl[\bigl(t+u(y-p)\bigr)\Ind_{\vartheta(t,y)\geq p}
    +u(y)\Ind_{\vartheta(t,y)<p}\Bigr] = \mathbb{E}\bigl[u(y)\bigr]
    +
    \mathbb{E}\Bigl[\bigl(t+u(y-p)-u(y)\bigr)\Ind_{\vartheta(t,y)\geq p}\Bigr].
\end{equation*}
Up to the constant first term, the contribution to welfare of a buyer with willingness to pay $\theta$ facing price $p$ is
\begin{equation*}
    w(\theta,p)
    \coloneqq
    \mathbb{E}\Bigl[t+u(y-p)-u(y)\Bigm\vert \vartheta(t,y)=\theta\Bigr]
    \Ind_{\theta\geq p}.
\end{equation*}
Using the indifference condition, conditional on $\vartheta(t,y)=\theta$, we obtain
\begin{equation*}
    t=u(y)-u(y-\theta).
\end{equation*}
Therefore,
\begin{equation*}
    w(\theta,p)
    =
    \mathbb{E}\Bigl[u(y-p)-u(y-\theta)\Bigm\vert \vartheta(t,y)=\theta\Bigr]
    \Ind_{\theta\geq p} \label{eq:A.2} \tag{A.2}.
\end{equation*}

This immediately implies the first two requirements imposed on $w$ in the
main text: if $\theta\leq p$, then $w(\theta,p)=0$, and, if $p<p'\leq \theta$, then
\begin{equation*}
    w(\theta,p)-w(\theta,p')
=
\mathbb E\!\left[
u(y-p)-u(y-p')\mid \vartheta(t,y)=\theta
\right]\geq 0.
\end{equation*}

Furthermore, $w$ satisfies the redistribution condition \eqref{cond:R} if the mapping
\begin{equation*}
    \theta\longmapsto w(\theta,p)-w(\theta,p') =\mathbb{E}\Bigl[u(y-p)-u(y-p')\Bigm\vert \vartheta(t,y)=\theta\Bigr],
\end{equation*}
is non-increasing. Given that $u$ is concave, $u(y-p)-u(y-p')$ is decreasing with $y$. Therefore, this mapping is non-increasing if larger willingness to pay is correlated with larger income. For instance, if $F(y \mid \theta)$ denotes the cdf of revenues conditional on willingness to pay\footnote{Note that this conditional cdf is determined by the joint distribution of $(t,y)$ and the map $\vartheta(t,y)$.}, then a sufficient condition for $w\in \mathcal{R}$ is that
\begin{equation*}
    \theta<\theta' \implies F(y\mid\theta')\leq F(y\mid\theta) \quad \forall y.
\end{equation*}


Conversely, since $\vartheta(t,y)$ is non-decreasing in both taste and income, a violation of \eqref{cond:R} can occur only when higher willingness-to-pay groups have a higher expected marginal utility of money than lower willingness-to-pay groups over some price interval. With concave $u$, this means that high willingness to pay must be driven sufficiently by high taste rather than by high income. 

\paragraph{Endogenizing the welfare weights}

Since $u$ is concave, and therefore absolutely continuous, we have
\begin{equation*}
    u(y-p)-u(y-\theta)
    =
    \int_{p}^{\theta} u'(y-s) \, ds.
\end{equation*}
for all $p<\theta$. Therefore, for $p<\theta$, we can write $w(\theta,p) = \bar{\lambda}(\theta,p)(\theta-p)$, where
\begin{equation*}
    \bar{\lambda}(\theta,p)=
    \mathbb{E}\left[
        \frac{1}{\theta-p}
        \int_{p}^{\theta} u'(y-s)\,ds
        \biggm\vert \vartheta(t,y)=\theta
    \right],
\end{equation*}
is an endogenous welfare weight on the buyer surplus \citep[this is reminiscent of the literature on optimal taxation, as in][]{saez}

When the transaction is small relative to income, or when $u'$ varies little over the relevant range of post-purchase incomes, this weight is well approximated by
\begin{equation*}
    \lambda(\theta)
    \coloneqq
    \mathbb{E}\bigl[
        u'(y)
        \bigm\vert \vartheta(t,y)=\theta
    \bigr].
\end{equation*}
In this case, $w(\theta,p) \approx \lambda(\theta)(\theta-p)\Ind_{\theta\geq p}$, which is the weighted-surplus specification used in \cref{ex:decreasing_pareto}, but the weights are no longer primitive. Similarly as before, since $u$ is concave, the weights $\lambda(\cdot)$ are non-increasing if higher willingness-to-pay types have higher conditional incomes. 

\section{Proof of \cref{prop:red_order}}\label{app:proof_red_order}

\begin{proof}
For any $f,g \in \mathbb{R}^\Omega$ we let $\langle f, g \rangle$ be defined as
\begin{equation*}
    \langle f, g \rangle = \sum_{(\theta,p)\in\Omega} f(\theta,p) \, g(\theta,p).
\end{equation*}

We know from \cref{lem:efficiency} that we can, without loss of generality, restrict attention to segmentations supported on $\Omega$. Hence, we treat a segmentation $\sigma$ as an element of $\mathbb{R}^\Omega$. 

We first define the following set of functions,\footnote{$\mathcal{E}$ is the set of real functions $w$ defined on $\Omega$ that is supermodular and such that, for any $\theta \in \Theta$, the function $p \mapsto w(\theta,p)$ is non-increasing on $\{\theta_1, \dots, \theta\}$.}
\begin{equation*}
    \mathcal{E} = \Bigl\{v \in \mathbb{R}^\Omega \ \big| \ \exists w \in \mathcal{R} \text{ and } b \in \mathbb{R} \text{ such that}, \forall (\theta,p)\in\Omega, \, v(\theta,p) = w(\theta,p) +b \Bigr\}.
\end{equation*}

For any given aggregate market $\mu$, we define the order $\succsim_{\mathcal{E}}$ over the set $\Sigma^{\star}(\mu)$ in the analogous way we defined the order $\succsim_\mathcal{R}$. The following lemma allows us to work with the family of functions $\mathcal{E}$.
\begin{lemma}\label{lem:extended_set_of_functions}
    $\succsim_{\mathcal{E}}=\succsim_\mathcal{R}$.
\end{lemma}
\begin{proof}
    As $\mathcal{R} \subset \mathcal{E}$, if $\sigma \succsim_\mathcal{E} \sigma'$, then $\sigma \succsim_\mathcal{R} \sigma'$. Conversely, let $\sigma, \sigma'$ be such that $\sigma \succsim_\mathcal{R}\sigma'$, that is, for any $w\in \mathcal{R}$, $\langle \sigma, w \rangle \geq \langle \sigma', w \rangle$. This implies that for any $w\in \mathcal{R}$ and $b \in \mathbb{R}$, $\langle \sigma, w \rangle + b \geq \langle \sigma', w \rangle +b$. With a slight abuse of notation, by denoting $b$ also as the element of $\mathbb{R}^{\Theta \times P}$ constantly equal to $b$, $\langle \sigma, b \rangle = \langle \sigma', b \rangle = b$. Therefore, for any $w\in \mathcal{R}$ and $b \in \mathbb{R}$,
    \begin{align*}
        \langle \sigma, w \rangle + b \geq \langle \sigma', w \rangle +b & \iff \langle \sigma, w \rangle + \langle \sigma, b \rangle \geq \langle \sigma', w \rangle + \langle \sigma', b \rangle \\
        & \iff \langle \sigma, w +b \rangle \geq \langle \sigma', w +b \rangle
    \end{align*}
    
    This implies that $\sigma \succsim_\mathcal{E} \sigma'$.
\end{proof}

We now prove the statement of the theorem for the order $\succsim_\mathcal{E}$. The proof follows similar steps as the proof of Theorem 1 in \cite{Meyer2015}. 
We first show that the set $\mathcal{E}$ is the dual cone of the set $\mathcal{T}$, as shown by the next lemma.
\begin{lemma}\label{lem:proof_red_order}
    For any $w\in \mathbb{R}^\Omega$,
    \begin{equation*}
        w\in \mathcal{E} \iff \forall t\in \mathcal{T}, \; \langle w, t \rangle \geq 0.
    \end{equation*}
\end{lemma}
\begin{proof}
    Necessity follows from the definition of $\mathcal{R}$ and therefore of $\mathcal{E}$.  Let us prove sufficiency. Let $w$ be such that $\langle w, t \rangle \geq 0$ for every $t\in \mathcal{T}$. By way of contradiction, assume that there exists $(\theta, p),(\theta, p')\in \Omega$ such that $p<p'$ and $w(\theta,p)< w(\theta, p')$. Consider the downward buyer transfer $d\in \mathbb{R}^\Omega$ such that $d(\theta,p)=-d(\theta,p')=1$. Then $\langle w, d \rangle < 0$, a contradiction. Assume now that there exists $(\theta, p),(\theta', p')\in \Omega$ such that $\theta < \theta'$, $p<p'$, and $w(\theta,p)-w(\theta,p') < w(\theta',p)-w(\theta',p')$. Consider the redistributive buyer transfer $r\in \mathbb{R}^\Omega$ such that
    \begin{equation*}
        r(\theta,p)=r(\theta',p')=-r(\theta',p)=-r(\theta,p')=1.
    \end{equation*}
    Then $\langle w, r \rangle < 0$, again a contradiction.
\end{proof}

Let $\cone(\mathcal{T})$ denote the conic hull of $\mathcal{T}$, the smallest convex cone containing $\mathcal{T}$, equal to the set of all positive linear combinations of finitely many elements of $\mathcal{T}$ \citep[see][p. 430]{ok2007real}:
\begin{equation*}
    \cone(\mathcal{T}) = \left\{\sum_{t\in T} \lambda_{t} t \ \Big| \  \text{$T\subset \mathcal{T}$, $\vert T\rvert < \infty$ and $\lambda \in \mathbb{R}_{+}^{T}$} \right\}.
\end{equation*}

Note that for any $t \in \mathcal{T}$ and $\lambda_{t} \geq 0$, $\lambda_{t} t \in \mathcal{T}$. Therefore,
\begin{equation*}
    \cone(\mathcal{T}) = \left\{\sum_{i=1}^n t_i \ \Big| \ \text{$n \in \mathbb{N}^{*}$ and $(t_i)_{i\in\{1,\dots,n\}} \in \mathcal{T}^n$} \right\}.
\end{equation*}

As $\cone(\mathcal{T})$ is closed, we can show in the next lemma that it is the dual cone of $\mathcal{E}$ \citep[see][p.215]{luenberger1969optimization}. 
\begin{lemma}\label{lem:proof_dual_cone}
    For any $\tau\in \mathbb{R}^\Omega$,
    \begin{equation*}
        \tau \in \cone(\mathcal{T}) \iff \forall w\in \mathcal{E}, \; \langle w, \tau\rangle \geq 0.
    \end{equation*}
\end{lemma}
\begin{proof}
    Assume that $\tau\in \cone(\mathcal{T})$. Hence, there exists $(t_i)_{i\in\{1,\dots,n\}} \in \mathcal{T}^n$ such that $\tau = \sum_{i=1}^{n} t_i$. Following \cref{lem:proof_red_order} we have that, for any $w\in \mathcal{E}$,
    \begin{equation*}
        \langle \tau, w \rangle= \sum_{i=1}^{n} \langle t_i, w\rangle \geq 0.
    \end{equation*}
    
    Assume now that $\upsilon\notin \cone(\mathcal{T})$. As $\cone(\mathcal{T})$ is closed, the Separating Hyperplane Theorem implies that there exists $w \in \mathbb{R}^{\Omega}$ such that $\langle \upsilon, w \rangle < \langle \tau, w \rangle$ for every $\tau \in \cone(\mathcal{T})$. Since $\cone(\mathcal{T})$ is a cone, it must be that $\langle \tau, w \rangle \geq 0$ for any $\tau\in\cone(\mathcal{T})$ as otherwise there would exist some $\alpha >0$ such that $\langle \alpha \tau , w \rangle \leq \langle \upsilon, w \rangle$. Thus, \cref{lem:proof_red_order} implies that $w \in \mathcal{E}$. Finally, since $0=\min_{\tau \in \cone(\mathcal{T})} \langle \tau, w \rangle $, we must have $\langle \upsilon, w\rangle<0$.
\end{proof}

We can now restate the theorem as follows: for any $\mu \in \Delta(\Theta)$ and $\sigma, \sigma' \in \Sigma^{\star}(\mu)$, $\sigma \succsim_{\mathcal{E}} \sigma'$ if and only if there exists $n\in\mathbb{N}^{*}$ and a sequence $(t_i)_{i\in\{1,\dots,n\}} \in \mathcal{T}^{n}$ such that:
\begin{align}
    \sigma = \sigma' + \sum_{i=1}^{n} t_i \label{eqn:proof_th1_1}
\end{align}

\Cref{eqn:proof_th1_1} is equivalent to $\sigma - \sigma' \in \cone(\mathcal{T})$, and by \cref{lem:proof_dual_cone}, this is equivalent to $\langle \sigma-\sigma', w \rangle \geq 0 $ for every $w \in \mathcal{E}$.
\end{proof}

\section{Proofs of \cref{prop:max_order_opt,prop:charac_max_order} in the Proof of \cref{thm:opt_red_seg}}

\subsection{Proof of \cref{prop:max_order_opt}}\label{app:proof_satur}

\begin{proof}[Proof of \cref{prop:max_order_opt}]
Claim (i) is direct: If $\sigma$ were not $\succsim_{\mathcal{R}}$-maximal, there would exist at least one segmentation yielding a weakly higher aggregate welfare for any redistributive welfare function, and for $w$ in particular. Furthermore, if $w\in \overline{\mathcal{R}}$, this segmentation would yield a strictly higher aggregate welfare.

We now prove claim (ii). For any set $A$ we let $-A=\bigl\{-a \ | \ a\in A\bigr\}$. For any two sets $A$ and $B$, we let $A+B\coloneqq\bigl\{a+b \ | \ a\in A, b\in B\bigr\}$ be the Minkowski sum of $A$ and $B$, and we let $A-B=A + (-B)$.

Note that, for any $\mu\in\Delta(\Theta)$, the set $\Sigma^{\star}(\mu)$ is a convex polytope. Let $\sigma \in \Sigma^{\star}(\mu)$ be $\succsim_{\mathcal{R}}$-maximal. \cref{prop:red_order} implies that for any $t \in \cone(\mathcal{T})$ such that $t\neq 0$, we must have $\sigma +t \notin \Sigma^{\star}(\mu)$. Therefore, $\bigl[\Sigma^{\star}(\mu) \cap \bigl(\{\sigma\} + \cone(\mathcal{T})\bigr)\bigr] =\{\sigma\}$. This implies that $\sigma$ is on the relative border of both $\Sigma^{\star}(\mu)$ and $\{\sigma\} + \cone(\mathcal{T})$. 

Given that these are closed convex sets and $\Sigma^\star(\mu)$ is compact, the separating hyperplane theorem implies that there exists a nonzero $v \in \mathbb{R}^\Omega$ such that:
\begin{align*}
    \forall \sigma' \in \Sigma^{\star}(\mu), \sigma'\neq \sigma, \forall t \in \cone(\mathcal{T}), t\neq 0, \quad 
    \langle \sigma', v \rangle < \langle \sigma, v \rangle < \langle \sigma + t, v \rangle.
\end{align*}

Furthermore, for any $t \in \cone(\mathcal{T})$, $\langle \sigma, v \rangle \leq \langle \sigma + t, v \rangle \iff 0\leq \langle t, v \rangle$, which, by \cref{lem:proof_red_order}, implies that $v\in \mathcal{E}$. Therefore, there exists $w\in \mathcal{R}$ and $b\in \mathbb{R}$ such that $v=w + b$. Note that $v \neq b$, as otherwise for any $t\in \mathcal{T}$, it would imply that $\langle t,v\rangle =0$, contradicting the previous inequality $\langle \sigma, v \rangle < \langle \sigma + t, v \rangle$. Therefore, $w\neq 0$. Hence, for any $\sigma' \in \Sigma^{\star}(\mu)$,
\begin{align*}
    \langle \sigma', v \rangle \leq \langle \sigma, v \rangle &\iff \langle \sigma', w +b \rangle \leq \langle \sigma,  w+b \rangle \\
    &\iff \langle \sigma', w \rangle \leq \langle \sigma, w \rangle,
\end{align*}
which proves that $\sigma$ is \labelcref{eqn:designer_pb_efficient}-optimal.
\end{proof}

\subsection{Proof of \cref{prop:charac_max_order}}\label{app:proof_satur_2}

\begin{proof}[Proof of \cref{prop:charac_max_order}]

For any $x\in\mathbb{R}^\Omega$ and any $p,q\in P$, write
\begin{equation*}
    \Pi_{p,q}(x)\coloneqq p\sum_{\theta\geq p}x(\theta,p)-q\sum_{\theta\geq q}x(\theta,p).
\end{equation*}
Thus, for any $\mu\in \Delta(\Theta)$ and $\sigma\in \Sigma^{\star}(\mu)$, the constraint $\Pi_{p,q}(\sigma)\geq 0$ is the obedience condition ensuring that, in segment $p$, the seller's profit from setting a price equal to $p$ is at least as high as from setting a price equal to $q$. For any $\mu\in \Delta(\Theta)$, $\sigma\in \Sigma^{\star}(\mu)$, and $p\in P$, let
\begin{equation*}
    \mathcal{I}_{\sigma}(p)\coloneqq \bigl\{q\in P \bigm\vert \Pi_{p,q}(\sigma)=0\bigr\},
\end{equation*}
be the set of prices that are tied with $p$ in segment $p$.

We begin by establishing that a segmentation is saturated if and only if perturbing it with any single downward or redistributive transfer leads to a violation of obedience.

\begin{lemma}\label{lem:saturation_iff_no_trans}
    For any $\sigma \in \Sigma^{\star}(\mu)$, $\sigma$ is saturated if and only if there exists no $t\in \mathcal{T}$ such that $t\neq 0$ and $\sigma + t \in \Sigma^{\star}(\mu)$.
\end{lemma}
\begin{proof}[Proof of \cref{lem:saturation_iff_no_trans}]
    \framebox{$\Leftarrow$} Suppose that there exists no $t \in \mathcal{T}$ such that $t\neq 0$ and $\sigma + t \in \Sigma^{\star}(\mu)$. Downward and redistributive buyer transfers preserve the marginal distribution over types, so any infeasibility created by such a transfer must come from a violation of obedience.

    We first prove part (a) of \cref{def:saturation}. Denote $\Bar{p} =\max \supp(\sigma_{P})$ and $\Bar{\theta}=\max \bigl\{\supp\bigl( \sigma(\, \cdot \, | \, \Bar{p})\bigr)\bigr\}$, and let $p \in \supp(\sigma_{P})$ be such that $p<\Bar{p}$. If there were no $q\in P$ with $p<q\leq \Bar{\theta}$ and $q\in\mathcal{I}_{\sigma}(p)$, then, for $\varepsilon>0$ sufficiently small, the downward transfer of mass $\varepsilon$ from $(\Bar{\theta},\Bar{p})$ to $(\Bar{\theta},p)$ would preserve obedience, a contradiction. Hence such a $q$ exists, and part (a) follows.

    We now prove part (b). Let $p,p'\in\supp(\sigma_P)$ with $p<p'$, and suppose by contradiction that there exists $q\in\supp\bigl(\sigma(\,\cdot\,|\,p)\bigr)$ such that $q\geq p'$ and $q\notin\mathcal{I}_{\sigma}(p')$. Since $p'\in\mathcal{I}_{\sigma}(p')$, we must have $q>p'$. Moreover, segment $p'$ must contain some type below $q$: otherwise charging $q$ would yield strictly more profit than charging $p'$ in segment $p'$. Let
    \begin{equation*}
        \theta_q\coloneqq\max \Bigl\{\theta\in\Theta \mid \theta<q,\; \theta\in\supp\bigl(\sigma(\,\cdot\,|\,p')\bigr)\Bigr\}.
    \end{equation*}
    For $\varepsilon>0$ sufficiently small, the redistributive transfer that moves mass $\varepsilon$ from $(\theta_q,p')$ and $(q,p)$ to $(\theta_q,p)$ and $(q,p')$ preserves obedience: it relaxes all constraints in segment $p$, and in segment $p'$ it tightens only constraints associated with prices in $(\theta_q,q]$, all of which are slack because $\theta_q$ is the largest type below $q$ in segment $p'$ and $q\notin\mathcal{I}_{\sigma}(p')$. This contradicts the assumption that no single transfer is feasible. Therefore part (b) holds.

    \framebox{$\Rightarrow$} Suppose that $\sigma \in \Sigma^{\star}(\mu)$ is saturated. We first show that no redistributive transfer is feasible. Suppose, to the contrary, that $\sigma'\in\Sigma^{\star}(\mu)$ is obtained from $\sigma$ by moving a mass $\varepsilon>0$ from $(\theta',p'')$ and $(\theta'',p')$ to $(\theta',p')$ and $(\theta'',p'')$, with $p'<p''$ and $\theta'<\theta''$. Since $(\theta'',p')\in\supp(\sigma)$, part (b) of \cref{def:saturation} gives $\theta''\in\mathcal{I}_{\sigma}(p'')$. The transfer then strictly tightens the obedience constraint comparing $p''$ to $\theta''$ in segment $p''$, contradicting $\sigma'\in\Sigma^{\star}(\mu)$.

    We next show that no downward transfer is feasible. Suppose, to the contrary, that $\sigma'\in\Sigma^{\star}(\mu)$ is obtained from $\sigma$ by moving a mass $\varepsilon>0$ of type $\theta'$ buyers from segment $p''$ to segment $p'$, with $(\theta',p'')\in\supp(\sigma)$ and $p'<p''$. By part (a) of \cref{def:saturation}, there exists $q>p'$ such that $q\in\mathcal{I}_{\sigma}(p')$. If $q\leq \theta'$, the transfer strictly tightens the constraint comparing $p'$ to $q$ in segment $p'$, so $\sigma'$ violates obedience. If instead $q>\theta'$, then $q\geq p''$ and part (b) of \cref{def:saturation} gives $q\in\mathcal{I}_{\sigma}(p'')$. In that case the transfer strictly tightens the constraint comparing $p''$ to $q$ in segment $p''$, again contradicting $\sigma'\in\Sigma^{\star}(\mu)$.
\end{proof}


We now show the local-to-global result.

\begin{lemma}\label{lem:local_global}
    Fix $\mu\in \Delta(\Theta)$. If $\sigma\in \Sigma^{\star}(\mu)$ is saturated, then for any $\tau \in \cone(\mathcal{T})$, $\tau\neq 0$, we have
    $\sigma + \tau \notin \Sigma^{\star}(\mu)$.
\end{lemma}

\begin{proof}[Proof of \cref{lem:local_global}]
We first characterize the extremal rays of $\cone(\mathcal{T})$. A ray of $\cone(\mathcal{T})$ is a set $\mathbb{R}_{+}\tau\coloneqq\{\lambda\tau\mid \lambda\in\mathbb{R}_{+}\}$ with $\tau\in\cone(\mathcal{T})$ and $\tau\neq 0$. It is \emph{extremal} if, whenever $\tau=\tau_1+\tau_2$ with $\tau_1,\tau_2\in\cone(\mathcal{T})$, we have $\tau_1,\tau_2\in\mathbb{R}_{+}\tau$.

Write $P=\{p_1,\dots,p_K\}$, where $p_i=\theta_i$ for every $i$. For $i\in\{1,\dots,K-1\}$, define the \emph{elementary downward transfer} $d_i\in \mathbb{R}^\Omega$ by
\begin{equation*}
    d_i(\theta,p)=
    \begin{cases}
        1 & \text{if }(\theta,p)=(\theta_K,p_i),\\
        -1 & \text{if }(\theta,p)=(\theta_K,p_{i+1}),\\
        0 & \text{otherwise.}
    \end{cases}
\end{equation*}
For $1\leq i<k\leq K-1$, define the \emph{elementary redistributive transfer} $r_{k,i}\in \mathbb{R}^\Omega$ by
\begin{equation*}
    r_{k,i}(\theta,p)=
    \begin{cases}
        1 & \text{if }(\theta,p)\in \bigl\{(\theta_k,p_i),(\theta_{k+1},p_{i+1})\bigr\},\\
        -1 & \text{if }(\theta,p)\in \bigl\{(\theta_k,p_{i+1}),(\theta_{k+1},p_i)\bigr\},\\
        0 & \text{otherwise.}
    \end{cases}
\end{equation*}
We refer to the index $i$ as the \emph{price step} associated with $d_i$ (respectively, $r_{k,i}$). That is to say that these elementary transfers reallocate buyers exclusively between the adjacent price segments $p_i$ and $p_{i+1}$.

\begin{lemma}\label{lem:extremal_rays_coneT}
    A ray of $\cone(\mathcal{T})$ is extremal if and only if it is generated by some elementary downward or redistributive transfer.
\end{lemma}

\begin{proof}
For $1\leq i<k\leq K$, let $d_{k,i}\in\mathbb{R}^{\Omega}$ denote the unit adjacent downward transfer of type $\theta_k$ from segment $p_{i+1}$ to segment $p_i$:
\begin{equation*}
    d_{k,i}(\theta,p)=
    \begin{cases}
        1 & \text{if }(\theta,p)=(\theta_k,p_i),\\
        -1 & \text{if }(\theta,p)=(\theta_k,p_{i+1}),\\
        0 & \text{otherwise.}
    \end{cases}
\end{equation*}
Then $d_i=d_{K,i}$ and $r_{k,i}=d_{k,i}-d_{k+1,i}$. Conversely, for $1\leq i<k\leq K-1$,
\begin{equation}\label{eq:adjacent_downward_decomp}
    d_{k,i}=d_i+\sum_{m=k}^{K-1} r_{m,i}.
\end{equation}

Note that every downward transfer of type $\theta_k$ from $p_j$ to $p_i$, with $i<j\leq k$, is a nonnegative multiple of
\begin{equation*}
    \sum_{\ell=i}^{j-1}d_{k,\ell},
\end{equation*}
and is therefore, by \eqref{eq:adjacent_downward_decomp}, a nonnegative combination of the elementary transfers. Likewise, every redistributive transfer swapping types $\theta_a<\theta_b$ between prices $p_i<p_j$, with $i<j\leq a<b$, is a nonnegative multiple of
\begin{equation*}
    \sum_{\ell=i}^{j-1}\bigl(d_{a,\ell}-d_{b,\ell}\bigr)
    =\sum_{\ell=i}^{j-1}\sum_{m=a}^{b-1}r_{m,\ell},
\end{equation*}
and hence is also a nonnegative combination of the elementary transfers. Thus $\cone(\mathcal{T})$ is generated by the set
\begin{equation*}
    \mathcal{G}\coloneqq \{d_i\mid 1\leq i\leq K-1\}\cup\{r_{k,i}\mid 1\leq i<k\leq K-1\}.
\end{equation*}

It remains to show that these generators are linearly independent. First, the vectors $\{d_{k,i}\mid 1\leq i<k\leq K\}$ are linearly independent: for each fixed type $\theta_k$, the vectors $d_{k,1},\dots,d_{k,k-1}$ have support only on row $\theta_k$, and the coefficient of $d_{k,1}$ is identified by the cell $(\theta_k,p_1)$, the coefficient of $d_{k,2}$ by the cell $(\theta_k,p_2)$ after subtracting the first one, and so on; rows corresponding to different types are disjoint.

Now suppose that
\begin{equation*}
    \sum_{i=1}^{K-1}\alpha_i d_i+\sum_{i=1}^{K-2}\sum_{k=i+1}^{K-1}\beta_{k,i}r_{k,i}=0.
\end{equation*}
Fix a price step $i$. Set $\beta_{i,i}=0$ and $\beta_{K,i}=\alpha_i$. When the left-hand side is written in the $d_{k,i}$'s, the coefficient of $d_{k,i}$ is $\beta_{k,i}-\beta_{k-1,i}$ for every $k=i+1,\dots,K$. Since the $d_{k,i}$'s are linearly independent, these coefficients must all be zero. Hence $\beta_{i+1,i}=\cdots=\beta_{K-1,i}=\alpha_i=0$. This holds for every $i$, so $\mathcal{G}$ is linearly independent.

Because $\mathcal{G}$ generates $\cone(\mathcal{T})$ and is linearly independent, every element of $\cone(\mathcal{T})$ has a unique representation as a nonnegative combination of elements of $\mathcal{G}$. Let $g\in\mathcal{G}$. If $g=\tau_1+\tau_2$ with $\tau_1,\tau_2\in\cone(\mathcal{T})$, uniqueness of the representation implies that all coefficients of $\tau_1$ and $\tau_2$ on generators other than $g$ are zero; hence $\tau_1,\tau_2\in\mathbb{R}_{+}g$. Therefore $\mathbb{R}_{+}g$ is an extremal ray. Conversely, if a nonzero $\tau\in\cone(\mathcal{T})$ has positive coefficients on at least two generators, its unique representation splits $\tau$ into two nonzero elements of the cone that are not proportional to $\tau$. Hence the ray generated by $\tau$ is not extremal. The extremal rays are therefore exactly the rays generated by the elementary transfers.
\end{proof}

We now prove \cref{lem:local_global}. By \cref{lem:extremal_rays_coneT}, any $\tau\in\cone(\mathcal{T})$ can be written uniquely as
\begin{equation*}
    \tau=\sum_{i=1}^{K-1}\left(\alpha_i d_i+\sum_{k=i+1}^{K-1}\beta_{k,i}r_{k,i}\right),
    \qquad \alpha_i,\beta_{k,i}\geq 0.
\end{equation*}
Consider a saturated $\sigma \in \Sigma^{\star}(\mu)$. Suppose, toward a contradiction, that $\tau\neq 0$ and $\sigma+\tau\in\Sigma^{\star}(\mu)$. Let $n$ be the lowest active price step, namely
\begin{equation*}
    n=\min\Bigl\{i\in\{1,\dots,K-1\}: \alpha_i>0 \text{ or } \beta_{k,i}>0 \text{ for some } k\in\{i+1,\dots,K-1\}\Bigr\}.
\end{equation*}
For $a\leq b$, write
\begin{equation*}
    \tau_{a:b}\coloneqq \sum_{i=a}^{b}\left(\alpha_i d_i+\sum_{k=i+1}^{K-1}\beta_{k,i}r_{k,i}\right).
\end{equation*}
We first find a type $q=\theta_j$, with $j\geq n+2$, such that
\begin{equation}\label{eq:local_start_deficit}
    \Pi_{p_{n+1},q}(\sigma+\tau_{n:n})<0
\end{equation}
and that
\begin{equation}\label{eq:local_binding_chain}
    \text{$ q\in \mathcal{I}_{\sigma}(p_i)$ for every $i\in\{n+1,\dots,j\}$.}
\end{equation}

First suppose that $\alpha_n=0$. Choose $j$ maximal such that $\beta_{j-1,n}>0$. The transfer $r_{j-1,n}$ removes mass from $(\theta_j,p_n)$. No lower price step is active, and by the maximality of $j$ no other elementary transfer at step $n$ adds mass to this cell. Since $\sigma+\tau$ is nonnegative, we must have $\sigma(\theta_j,p_n)>0$. Part (b) of \cref{def:saturation} therefore implies \eqref{eq:local_binding_chain} for every nonempty segment $p_i$ with $i\in \{ n+1,\dots,j \}$. If such a segment is empty, then $\Pi_{p_i,q}(\sigma)=0$ holds trivially. Among all elementary transfers at step $n$, only $r_{j-1,n}$ affects the constraint comparing $p_{n+1}$ to $q=\theta_j$. Hence
\begin{equation*}
    \Pi_{p_{n+1},q}(\sigma+\tau_{n:n})
    =\Pi_{p_{n+1},q}(\sigma)-\beta_{j-1,n}q
    =-\beta_{j-1,n}q<0,
\end{equation*}
which gives \eqref{eq:local_start_deficit}.

Now suppose that $\alpha_n>0$. If $\Pi_{p_n,p_{n+1}}(\sigma)=0$, then no redistributive transfer at step $n$ affects the constraint comparing $p_n$ to $p_{n+1}$, while $d_n$ tightens it by $\alpha_n(p_{n+1}-p_n)$. Since no later price step affects segment $p_n$, $\sigma+\tau$ would violate obedience. Hence $\Pi_{p_n,p_{n+1}}(\sigma)>0$, so segment $p_n$ is nonempty.

Segment $p_n$ is not the highest nonempty segment. Indeed, suppose that all segments above $p_n$ are empty under $\sigma$. Starting from $p_K$ and moving downward, nonnegativity of $\sigma+\tau$ forces all coefficients at price steps $K-1,K-2,\dots,n$ to be zero: once steps $m+1,\dots,K-1$ have been shown to be zero, the empty segment $p_{m+1}$ is affected only by step $m$, and nonnegativity at the cells $(\theta_{m+1},p_{m+1}),\dots,(\theta_K,p_{m+1})$ successively gives $\beta_{m+1,m}=\cdots=\beta_{K-1,m}=\alpha_m=0$. This contradicts the definition of $n$.

Thus $p_n<\max\supp(\sigma_P)$. By part (a) of \cref{def:saturation}, there exists $q=\theta_j>p_n$ such that $q\in\mathcal{I}_{\sigma}(p_n)$. Since $\Pi_{p_n,p_{n+1}}(\sigma)>0$, this binding price cannot be $p_{n+1}$, so $j\geq n+2$. Because $q$ is optimal in the nonempty segment $p_n$, we have $\sigma(\theta_j,p_n)>0$: otherwise either the next higher type would yield strictly larger profit than $q$, or the profit from $q$ would be zero. Part (b) of \cref{def:saturation} then gives \eqref{eq:local_binding_chain}.

Because $\sigma+\tau$ satisfies obedience in segment $p_n$, and because price steps above $n$ do not affect that segment, step $n$ must repair the loss caused by $d_n$ in the constraint comparing $p_n$ to $q$. The only elementary transfer at step $n$ that relaxes this constraint is $r_{j-1,n}$, and therefore
\begin{equation*}
    \beta_{j-1,n}q-\alpha_n(q-p_n)\geq 0.
\end{equation*}
In particular, $\beta_{j-1,n}>0$. The effect of step $n$ on the constraint comparing $p_{n+1}$ to $q$ is
\begin{align*}
    \Pi_{p_{n+1},q}(\sigma+\tau_{n:n})
    &=\alpha_n(q-p_{n+1})-\beta_{j-1,n}q \\
    &\leq \alpha_n(q-p_{n+1})-\alpha_n(q-p_n) \\
    &=\alpha_n(p_n-p_{n+1})<0,
\end{align*}
where we used $\Pi_{p_{n+1},q}(\sigma)=0$ from \eqref{eq:local_binding_chain}. This proves \eqref{eq:local_start_deficit} in the second case as well.

It remains to propagate the deficit from segment $p_{n+1}$ up to segment $p_{j-1}$. For each $i\in\{n+1,\dots,j-1\}$, define
\begin{equation*}
    \Pi_{i}\coloneqq \Pi_{p_i,q}(\sigma+\tau_{n:i-1}).
\end{equation*}
By \eqref{eq:local_start_deficit}, $\Pi_{n+1}<0$. We claim that $\Pi_{i}<0$ for every $i\in \{n+1,\dots,j-1\}$.

Assume $i\leq j-2$ and $\Pi_{i}<0$. Since $\sigma+\tau$ is obedient and no price step above $i$ affects segment $p_i$, step $i$ must repair this deficit. At step $i$, the downward transfer $d_i$ tightens the constraint comparing $p_i$ to $q$ by $\alpha_i(q-p_i)$, while the only elementary transfer that relaxes it is $r_{j-1,i}$. Thus
\begin{equation*}
    \Pi_{i}-\alpha_i(q-p_i)+\beta_{j-1,i}q\geq 0,
\end{equation*}
or equivalently
\begin{equation}\label{eq:repair_bound}
    \beta_{j-1,i}q\geq -\Pi_{i}+\alpha_i(q-p_i).
\end{equation}
Now consider segment $p_{i+1}$. Price steps $n,\dots,i-1$ do not affect this segment, and $q\in\mathcal{I}_{\sigma}(p_{i+1})$ by \eqref{eq:local_binding_chain}; hence $\Pi_{p_{i+1},q}(\sigma+\tau_{n:i-1})=0$. Therefore
\begin{align*}
    \Pi_{i+1}
    &=\Pi_{p_{i+1},q}(\sigma+\tau_{n:i}) \\
    &=\alpha_i(q-p_{i+1})-\beta_{j-1,i}q \\
    &\leq \alpha_i(q-p_{i+1})-\bigl[-\Pi_{i}+\alpha_i(q-p_i)\bigr] \\
    &=\Pi_{i}+\alpha_i(p_i-p_{i+1})<0,
\end{align*}
where the weak inequality uses \eqref{eq:repair_bound}. This proves the claim.

Taking $i=j-1$ in the claim gives
\begin{equation*}
    \Pi_{p_{j-1},q}(\sigma+\tau_{n:j-2})<0.
\end{equation*}
Step $j-1$ cannot repair this final deficit. The elementary redistributive transfer $r_{j-1,j-1}$ does not exist, while $d_{j-1}$ only tightens the constraint comparing $p_{j-1}$ to $q=\theta_j$. Price steps above $j-1$ do not affect segment $p_{j-1}$. Hence $\sigma+\tau$ violates obedience in segment $p_{j-1}$, contradicting $\sigma+\tau\in\Sigma^{\star}(\mu)$.

Therefore no nonzero $\tau\in\cone(\mathcal{T})$ satisfies $\sigma+\tau\in\Sigma^{\star}(\mu)$.
\end{proof}

We can now conclude the proof of \cref{prop:charac_max_order}. If $\sigma$ is $\succsim_{\mathcal{R}}$-maximal, then by \cref{prop:red_order} there is no nonzero $t\in\mathcal{T}$ such that $\sigma+t\in\Sigma^{\star}(\mu)$. Hence $\sigma$ is saturated by \cref{lem:saturation_iff_no_trans}. Conversely, if $\sigma$ is saturated, \cref{lem:local_global} rules out every nonzero $\tau\in\cone(\mathcal{T})$ with $\sigma+\tau\in\Sigma^{\star}(\mu)$. By \cref{prop:red_order}, no feasible segmentation is strictly above $\sigma$ in the redistributive order, so $\sigma$ is $\succsim_{\mathcal{R}}$-maximal.

\end{proof}
\section{Proof of \cref{prop:strong_mon}}\label{app:proof_mon}

\begin{proof}[Proof of \cref{prop:strong_mon}]
\framebox{$\Leftarrow$} Consider a segmentation $\sigma \in \Sigma^{\star}(\mu)$ that is $\succsim_{\mathcal{R}}$-maximal but not strongly monotone. Thus, there exist at least two segments $p, \theta_i \in \supp(\sigma_{P})$ such that $\max \supp\bigl(\sigma(\, \cdot \, | \, p)\bigr)>\theta_i>p$. We denote as $\theta'$ the lowest type in the segment $p$ that is strictly above $\theta_i$, and as $\theta_{k}$ the maximal type in the segment $\theta_i$: 
\begin{align*}
    & \theta' = \min \bigl\{x \in \supp\bigl(\sigma(\, \cdot \, | \, p)\bigr) \ | \ x> \theta_i \bigr\}, \\
    & \theta_{k} = \max \supp\bigl(\sigma(\, \cdot \, | \, \theta_i)\bigr).
\end{align*}

Let $\varepsilon>0$ be sufficiently small and $r$ be the redistributive buyer transfer such that
\begin{equation*}
    r(\theta_i,p) = r(\theta',\theta_i) = - r(\theta_i,\theta_i) = - r(\theta',p) = \varepsilon,
\end{equation*}
and $r = 0$ otherwise. Neglecting for the moment the negative indirect impacts of $r$ by triggering a violation of the obedient constraint in the segment $\theta_i$, the direct impact on the welfare of types $\theta_i$ and $\theta'$ involved in the transfer $r$ is
\begin{align}
    & \varepsilon\left[(w(\theta_i, p) - w(\theta_i, \theta_i))- \left(w(\theta', p)-w(\theta', \theta_i)\right)\right] \label{eqn:gain} \\
    > & \varepsilon\left[(w(\theta_i, p) - w(\theta_i, \theta_i))- \left(w(\theta_{i+1}, p)-w(\theta_{i+1}, \theta_i)\right)\right]> 0, \label{eqn:gain-2}
\end{align} 
where both strict inequalities come from $w \in \overline{\mathcal{R}}$. In addition, given that $\sigma$ is $\succsim_{\mathcal{R}}$-maximal, $\sigma + r$ does not satisfy \labelcref{eqn:obedience}. In particular, there exists $\hat{\theta}$, with $\theta_i < \hat{\theta} \leq \theta'$, such that $\hat{\theta}\in \mathcal{I}_{\sigma}(\theta_i)$ (see \cref{app:proof_satur_2} for the definition of $\mathcal{I}_{\sigma}(\theta)$), and thus:
\begin{equation*}
    \theta \sum_{x \geq \theta} \sigma(x,\theta) < \hat{\theta} \left(\sum_{x \geq \hat{\theta}} \sigma(x,\theta) + \varepsilon\right).
\end{equation*}

To recover the obedience of the price recommendation $\theta_i$, we implement the transfer $t\colon \Omega \to \mathbb{R}$, which transfers a mass $\frac{\theta_{i+1}}{\theta_{i+1}-\theta_i}\varepsilon$ of buyers from $(\theta_{k},\theta_i)$ to $(\theta_{k},\theta_{k})$: that is,
\begin{equation*}
    t(\theta_{k}, \theta_{k}) = - t(\theta_{k},\theta_i) = \frac{\theta_{i+1}}{\theta_{i+1}-\theta_i}\varepsilon,
\end{equation*}
and $t=0$ otherwise. First, for sufficiently small $\varepsilon$, $\sigma + t \in \Delta(\Omega)$. Second, implementing $t$ cannot break any obedience constraints, and it corrects from any violation of obedience constraints induced by the transfer $r$. To see this, look at the difference between the profit made by setting the price $\theta_i$ and setting the price $q \geq \theta_{i+1}$:
\begin{align*}
    & \theta_i \left(\sum_{x \geq \theta_i} \sigma(x,\theta_i) - \frac{\theta_{i+1}}{\theta_{i+1}-\theta_i}\varepsilon\right) - q \left(\sum_{x \geq q} \sigma(x,\theta_i) + \varepsilon - \frac{\theta_{i+1}}{\theta_{i+1}-\theta_i}\varepsilon\right) \\
    \geq & \varepsilon \left(\frac{\theta_{i+1}}{\theta_{i+1}-\theta_i} (q - \theta_i) -q\right)=\varepsilon\frac{\theta_i(q-\theta_{i+1})}{\theta_{i+1}-\theta_i} \geq 0
\end{align*}
where the first inequality follows from $\theta_i$ being an optimal price in segment $\theta_i$ according to $\sigma$, and the second inequality follows from $q \geq \theta_{i+1}$.

Therefore, $\sigma + r + t \in \Sigma^{\star}(\mu)$. The loss from the transfer $t$ is 
\begin{equation}\label{eqn:loss}
    \frac{\theta_{i+1}}{\theta_{i+1}-\theta}\varepsilon \bigr[w(\theta_{k},\theta_i) - w(\theta_{k}, \theta_{k})\bigr].
\end{equation}

Therefore, implementing the transfers $r$ and $t$ is beneficial if and only if \labelcref{eqn:gain} $\geq$ \labelcref{eqn:loss}. If $w$ is strongly redistributive, we have that \labelcref{eqn:gain-2} $\geq$ \labelcref{eqn:loss}. Given that \labelcref{eqn:gain} $\geq$ \labelcref{eqn:gain-2}, this ends the proof.

\framebox{$\Rightarrow$} We prove the contrapositive. Suppose that $w \in \overline{\mathcal{R}}$ is not strongly redistributive. So there exist $p <  \theta_i < \theta_{i+1} \leq \theta_{k}$, such that
\begin{equation}\label{eqn:violate-SR}
    \begin{aligned}
        \bigl[w(\theta_i, p) - w(\theta_i,\theta_i)\bigr] -&\bigl[w(\theta_{i+1}, p)-w(\theta_{i+1}, \theta_{i})\bigr] \\[5pt] 
        &< \frac{\theta_{i+1}}{\theta_{i+1}-\theta_{i}}\bigl[w(\theta_{k},\theta_{i}) -  w(\theta_{k},\theta_{k})\bigr].
    \end{aligned}
\end{equation}
    
Let $\mu \in \Delta\bigl(\{p, \theta_i, \theta_{i+1}, \theta_{k}\}\bigr)$ such that the unique strongly monotone and $\succsim_\mathcal{R}$-maximal segmentation, denoted $\sigma$, has three segments:
\begin{itemize}
    \item one segment $p$ including all types $p$ and some type $\theta_i$;
    \item one segment $\theta_i$ including the remaining types $\theta_i$ and all types $\theta_{i+1}$;
    \item and one segment $\theta_{k}$ including all types $\theta_{k}$.
\end{itemize}
Such a $\mu$ exists. For instance, let $\mu(p)>0$ and $\mu(\theta_{i+1})>0$ be arbitrarily small, and set $\mu(\theta_i) = \frac{p}{\theta_i-p} \mu(p) + \frac{\theta_{i+1}-\theta_i}{\theta_i}\mu(\theta_{i+1})$ and $\mu(\theta_k) = 1-\mu(p)-\mu(\theta_i)-\mu(\theta_{i+1})$. One can verify that this $\mu$ satisfies all the desired conditions. Let $\varepsilon>0$ be sufficiently small. We will now perform two transfers. The first one is a downward transfer $d$ defined by:
\begin{equation*}
    d(\theta_{k}, \theta_i)= - d(\theta_{k}, \theta_{k})=\frac{\theta_{i+1}}{\theta_{i+1}-\theta_i}\varepsilon,
\end{equation*}
$d=0$ otherwise. The second transfer, $s$, is the reverse of a redistributive one. It is defined by:
\begin{equation*}
    s(\theta_i,\theta_i) = s(\theta_{i+1},p) = - s(\theta_{i+1},\theta_i) = - s(\theta_i, p) = \varepsilon,
\end{equation*}
and $s=0$ otherwise. We check that $\sigma+d+s$ satisfies \labelcref{eqn:obedience}. First, for sufficiently small $\varepsilon$, this does not change the optimality of price $p$ in segment $p$. Second, in segment $\theta_i$, the transfers do not change the obedience constraint between for $\theta_{i+1}$ as
\begin{align*}
    \theta_i\left(\sigma(\theta_i, \theta_i) + \sigma(\theta_{i+1}, \theta_{i}) + \frac{\theta_{i+1}}{\theta_{i+1}-\theta_i}\varepsilon\right) &- \theta_{i+1}\left( \sigma(\theta_{i+1}, \theta_{i}) -\varepsilon + \frac{\theta_{i+1}}{\theta_{i+1}-\theta_i}\varepsilon\right) \\
    = & \varepsilon \left(\frac{\theta_i \theta_{i+1}}{\theta_{i+1}-\theta_1} - \frac{\theta_i \theta_{i+1}}{\theta_{i+1}-\theta_1}\right) \\
    = & 0,
\end{align*}
where the first equality follows from $\theta_{i+1}$ being an optimal price in segment $\theta_i$ for segmentation $\sigma$. Similarly, one can show that charging $\theta_i$ always gives a greater profit than charging $\theta_{k}$ if the segmentation is $\sigma + d + s$. Therefore $\sigma+d+s \in \Sigma^{\star}(\mu)$. Furthermore, it directly follows from \labelcref{eqn:violate-SR} that $\langle w, \sigma + d + s \rangle > \langle w, \sigma \rangle$. Therefore, $\sigma$, which is the unique strongly monotone and $\succsim_\mathcal{R}$-maximal segmentation of $\mu$ (\cref{lem:unique_strongly_mon}) is not optimal for $w$.
\end{proof}

\section{Proof of \cref{lem:unique_strongly_mon}}\label{app:proof_rent}

\begin{proof}
    First, if the uniform price of $\mu$ is $\theta_1$, there exists a unique segmentation that is saturated, the one that does not segment the market. Suppose now that the uniform price of $\mu$ is strictly greater than $\theta_1$. Let $\sigma \in \Sigma^{\star}(\mu)$ be strongly monotone and saturated. Denote as $\theta_{k}$ the maximal type that is in the segment $\theta_1$. By \cref{def:saturation}, there exists $\theta_i$, with $\theta_1 < \theta_i \leq \theta_{k}$ such that $\theta_i \in \mathcal{I}_\sigma(\theta_1)$. Consequently, any segmentation that is strongly monotone and saturated must have the same market segment $\theta_1$. We can conclude by iterating the same argument for the other segments.
\end{proof}

\bibliographystyle{ecta}
\bibliography{Biblio}

@article{aguirre2010monopoly,
    author = {Aguirre, Iñaki and Cowan, Simon and Vickers, John},
    title = {Monopoly Price Discrimination and Demand Curvature},
    journal = {American Economic Review},
    volume = {100},
    number = {4},
    pages = {1601–15},
    year = {2010},
    month = {September}
}

@article{akbarpour2024redistributive,
    author = {Akbarpour, Mohammad and Dworczak, Piotr and Kominers, Scott Duke},
    title = {Redistributive Allocation Mechanisms},
    journal = {Journal of Political Economy},
    volume = {132},
    number = {6},
    pages = {1831-1875},
    year = {2024},
    nameorder = {random}
}

@article{akbarpour2024vaccines,
    author = {Akbarpour, Mohammad and Budish, Eric and Dworczak, Piotr and Kominers, Scott Duke},
    title = {An Economic Framework for Vaccine Prioritization},
    journal = {Quarterly Journal of Economics},
    volume = {139},
    number = {1},
    pages = {359--417},
    year = {2024},
    nameorder = {random}
}

@article{arya2022rethinking,
    title = {Rethinking Distribution: Introducing Market Segmentation as a Policy Instrument},
    author = {Arya, Y. and Malhotra, R.},
    year = {2022},
    journal = {Working Paper}
}

@book{aumann1995,
    title = {Repeated Games with Incomplete Information},
    author = {Aumann, Robert J. and Maschler, Michael and Stearns, Richard E.},
    year = {1995},
    publisher = {MIT Press}
}

@article{bbm15,
    author = {Bergemann, Dirk and Brooks, Benjamin and Morris, Stephen},
    title = {The Limits of Price Discrimination},
    journal = {American Economic Review},
    volume = {105},
    number = {3},
    year = {2015},
    month = {March},
    pages = {921-57}
}

@techreport{bergemann2024alignment,
    title = {On the Alignment of Consumer Surplus and Total Surplus under Competitive Price Discrimination},
    author = {Bergemann, Dirk and Brooks, Benjamin and Morris, Stephen},
    year = {2024},
    month = {May},
    institution = {Cowles Foundation},
    type = {Discussion Paper},
    number = {2373R1}
}

@article{bergemann2019,
    title = {Information Design: A Unified Perspective},
    author = {Bergemann, Dirk and Morris, Stephen},
    journal = {Journal of Economic Literature},
    volume = {57},
    number = {1},
    pages = {44--95},
    year = {2019}
}

@article{bourreau2018regulation,
    title = {The Regulation of Personalised Pricing in the Digital Era},
    author = {Bourreau, Marc and De Streel, Alexandre},
    year = {2018},
    publisher = {DAF/COMP/WD (2018) 150},
    journal = {OECD Report}
}

@techreport{buchholz2020value,
    title = {The Value of Time: Evidence from Auctioned Cab Rides},
    author = {Buchholz, Nicholas and Doval, Laura and Kastl, Jakub and Matějka, Filip and Salz, Tobias},
    year = {2020},
    institution = {National Bureau of Economic Research}
}

@article{che2021,
    author = {Che, Yeon-Koo and Kim, Jinwoo and Kojima, Fuhito},
    title = {Weak Monotone Comparative Statics},
    journal = {arXiv preprint arXiv:1911.06442},
    year = {2021}
}

@article{condorelli2013,
    title = {Market and Non-Market Mechanisms for the Optimal Allocation of Scarce Resources},
    author = {Condorelli, Daniele},
    journal = {Games and Economic Behavior},
    volume = {82},
    pages = {582-591},
    year = {2013}
}

@article{cowan2016welfare,
    author = {Cowan, Simon},
    title = {Welfare-Increasing Third-Degree Price Discrimination},
    journal = {The RAND Journal of Economics},
    volume = {47},
    number = {2},
    pages = {326-340},
    year = {2016}
}

@article{doval_smolin_22,
    author = {Doval, Laura and Smolin, Alex},
    title = {Persuasion and Welfare},
    journal = {Journal of Political Economy},
    volume = {132},
    number = {7},
    pages = {2451-2487},
    year = {2024}
}

@article{dube2023,
    author = {Dubé, Jean-Pierre and Misra, Sanjog},
    title = {Personalized Pricing and Consumer Welfare},
    journal = {Journal of Political Economy},
    volume = {131},
    number = {1},
    pages = {131--189},
    year = {2023}
}

@article{dworczak21,
    author = {Dworczak, Piotr and Kominers, Scott Duke and Akbarpour, Mohammad},
    title = {Redistribution Through Markets},
    journal = {Econometrica},
    year = {2021},
    volume = {89},
    number = {4},
    pages = {1665-1698},
    month = {July},
    nameorder = {random}
}

@article{dworczak2024inequality,
    title = {Inequality and Market Design},
    author = {Dworczak, Piotr},
    journal = {ACM SIGecom Exchanges},
    year = {2024},
    month = {June},
    volume = {22},
    number = {1},
    pages = {83-92}
}

@article{elliott2024market,
    author = {Elliott, Matthew and Galeotti, Andrea and Koh, Andrew and Li, Wenhao},
    title = {Market Segmentation Through Information},
    journal = {Working Paper},
    year = {2024}
}

@book{galichon2018,
    title = {Optimal Transport Methods in Economics},
    author = {Galichon, Alfred},
    year = {2018},
    publisher = {Princeton University Press}
}

@article{galperti2023,
    author = {Galperti, Simone and Levkun, Aleksandr and Perego, Jacopo},
    title = {The Value of Data Records},
    journal = {Review of Economic Studies},
    volume = {91},
    number = {2},
    pages = {1007-1038},
    year = {2023},
    month = {April}
}

@article{hagpanah_siegel_aer,
    title = {The Limits of Multiproduct Price Discrimination},
    author = {Haghpanah, Nima and Siegel, Ron},
    journal = {American Economic Review: Insights},
    volume = {4},
    number = {4},
    pages = {443--458},
    year = {2022}
}

@article{hagpanah_siegel_jpe,
    author = {Haghpanah, Nima and Siegel, Ron},
    title = {Pareto-Improving Segmentation of Multiproduct Markets},
    journal = {Journal of Political Economy},
    volume = {131},
    number = {6},
    pages = {1546–1575},
    year = {2023},
    month = {June}
}

@article{ivanov2021,
    title = {Optimal Monotone Signals in Bayesian Persuasion Mechanisms},
    author = {Ivanov, Maxim},
    journal = {Economic Theory},
    volume = {72},
    number = {3},
    pages = {955--1000},
    year = {2021}
}

@article{kang2023,
    title = {The Public Option and Optimal Redistribution},
    author = {Kang, Zi Yang},
    year = {2023},
    journal = {Working Paper}
}

@article{kang2024a,
    title = {Optimal In-Kind Redistribution},
    author = {Kang, Zi Yang and Watt, Mitchell},
    journal = {arXiv preprint arXiv:2409.06112},
    year = {2024}
}

@article{kang2024b,
    title = {Optimal Redistribution Through Subsidies},
    author = {Kang, Zi Yang and Watt, Mitchell},
    journal = {Working Paper},
    year = {2024}
}

@article{kamenica_gentzkow11,
    author = {Kamenica, Emir and Gentzkow, Matthew},
    title = {Bayesian Persuasion},
    journal = {American Economic Review},
    volume = {101},
    number = {6},
    year = {2011},
    month = {October},
    pages = {2590-2615}
}

@article{kamenicareview,
    author = {Kamenica, Emir},
    title = {Bayesian Persuasion and Information Design},
    journal = {Annual Review of Economics},
    volume = {11},
    number = {1},
    pages = {249-272},
    year = {2019}
}

@article{Kolotilin2022,
author = {Kolotilin, Anton and Mylovanov, Timofiy and Zapechelnyuk, Andriy},
title = {Censorship as optimal persuasion},
journal = {Theoretical Economics},
volume = {17},
number = {2},
pages = {561-585},
doi = {https://doi.org/10.3982/TE4071},
url = {https://onlinelibrary.wiley.com/doi/abs/10.3982/TE4071},
eprint = {https://onlinelibrary.wiley.com/doi/pdf/10.3982/TE4071},
year = {2022}
}

@article{kolotilin2024,
    title = {On Monotone Persuasion},
    author = {Kolotilin, Anton and Li, Hongyi and Zapechelnyuk, Andy},
    journal = {Working Paper},
    year = {2024}
}

@article{legrand,
    title = {Public Price Discrimination and Aid to Low Income Groups},
    author = {Le Grand, Julian},
    journal = {Economica},
    volume = {42},
    number = {165},
    pages = {32--42},
    year = {1975}
}

@article{linliu,
    author = {Lin, Xiao and Liu, Ce},
    title = {Credible Persuasion},
    journal = {Journal of Political Economy},
    volume = {132},
    number = {7},
    pages = {2228-2273},
    year = {2024}
}

@book{luenberger1969optimization,
    title = {Optimization by Vector Space Methods},
    author = {Luenberger, David G.},
    year = {1969},
    publisher = {John Wiley \& Sons},
    address = {New York}
}

@article{meyer2015,
    author = {Meyer, Margaret and Strulovici, Bruno},
    title = {Beyond Correlation: Measuring Interdependence Through Complementarities},
    journal = {Working Paper},
    year = {2015}
}

@article{mensch2021,
    title = {Monotone Persuasion},
    author = {Mensch, Jeffrey},
    journal = {Games and Economic Behavior},
    volume = {130},
    pages = {521--542},
    year = {2021}
}

@article{mirrlees1971exploration,
    title = {An Exploration in the Theory of Optimum Income Taxation},
    author = {Mirrlees, James A.},
    journal = {Review of Economic Studies},
    volume = {38},
    number = {2},
    pages = {175--208},
    year = {1971}
}

@article{mohammed17how,
  title        = {How Retailers Use Personalized Prices to Test What You’re Willing to Pay},
  author       = {Mohammed, Rafi},
  journal      = {Harvard Business Review},
  year         = {2017},
  month        = oct,
  url          = {https://hbr.org/2017/10/how-retailers-use-personalized-prices-to-test-what-youre-willing-to-pay}
}

@inbook{muller2013,
    author = {M{\"u}ller, Alfred},
    title = {Stochastic Orders in Reliability and Risk: In Honor of Professor Moshe Shaked},
    publisher = {Springer},
    editors = {Haijun, Li and Xiaohu, Li},
    series = {Lecture Notes in Statistics},
    year = {2013},
    volume = {208},
    chapter = {2}
}

@report{oecd2022personalised,
    title = {Personalised Pricing in the Digital Era},
    author = {{OECD}},
    institution = {Organisation for Economic Co-operation and Development},
    year = {2022},
    month = {February},
    note = {Last accessed: October 5, 2024}
}

@book{ok2007real,
    title = {Real Analysis with Economic Applications},
    author = {Ok, Efe A.},
    year = {2007},
    publisher = {Princeton University Press},
    address = {Princeton, NJ}
}

@article{onuchic2023,
    title = {Conveying Value via Categories},
    author = {Onuchic, Paula and Ray, Debraj},
    journal = {Theoretical Economics},
    volume = {18},
    number = {4},
    pages = {1407--1439},
    year = {2023}
}

@article{pai2022taxing,
    title = {Taxing Externalities Without Hurting the Poor},
    author = {Pai, Mallesh and Strack, Philipp},
    journal = {Available at SSRN 4180522},
    year = {2025}
}

@book{pigou,
    author = {Pigou, Arthur Cecil},
    title = {The Economics of Welfare},
    publisher = {London: Macmillan},
    year = {1920}
}

@article{puccetti2015,
    author = {Puccetti, Giovanni and Wang, Ruodu},
    title = {Extremal Dependence Concepts},
    journal = {Statistical Science},
    volume = {30},
    number = {4},
    pages = {485--517},
    year = {2015}
}

@article{ray2018,
    title = {Certified Random: A New Order for Coauthorship},
    author = {Ray, Debraj and Robson, Arthur},
    journal = {American Economic Review},
    volume = {108},
    number = {2},
    pages = {489--520},
    year = {2018},
    nameorder = {random}
}

@article{reuter2020,
    title = {Mechanism Design for Unequal Societies},
    author = {Reuter, Marco and Groh, Carl-Christian},
    journal = {Available at SSRN 3688376},
    year = {2020}
}

@article{rhodeszhou,
    author = {Rhodes, Andrew and Zhou, Jidong},
    title = {Personalized Pricing and Competition},
    journal = {American Economic Review},
    volume = {114},
    number = {7},
    year = {2024},
    month = {July},
    pages = {2141–70}
}

@book{robinson,
    author = {Robinson, Joan},
    title = {The Economics of Imperfect Competition},
    publisher = {London: Macmillan},
    year = {1933}
}

@article{romanyuk2019,
    title = {Cream Skimming and Information Design in Matching Markets},
    author = {Romanyuk, Gleb and Smolin, Alex},
    journal = {American Economic Journal: Microeconomics},
    volume = {11},
    number = {2},
    pages = {250--276},
    year = {2019}
}

@article{rothschild1976,
    title = {Equilibrium in Competitive Insurance Markets},
    author = {Rothschild, Michael and Stiglitz, Joseph},
    journal = {Quarterly Journal of Economics},
    pages = {629--649},
    year = {1976}
}

@article{saez,
Author = {Saez, Emmanuel and Stantcheva, Stefanie},
Title = {Generalized Social Marginal Welfare Weights for Optimal Tax Theory},
Journal = {American Economic Review},
Volume = {106},
Number = {1},
Year = {2016},
Month = {January},
Pages = {24–45},
DOI = {10.1257/aer.20141362},
URL = {https://www.aeaweb.org/articles?id=10.1257/aer.20141362}
}

@article{schmalensee,
    author = {Schmalensee, Richard},
    title = {Output and Welfare Implications of Monopolistic Third-Degree Price Discrimination},
    journal = {American Economic Review},
    volume = {71},
    number = {1},
    pages = {242--247},
    year = {1981}
}

@article{schlom2024,
    title = {Price Distribution Regulation},
    author = {Schlom, Christoph},
    journal = {Working Paper},
    year = {2024}
}

@book{shaked2007,
    title = {Stochastic Orders},
    author = {Shaked, Moshe and Shanthikumar, J. George},
    publisher = {Springer},
    series = {Springer Series in Statistics},
    year = {2007}
}

@article{terstiege2023,
    title = {Market Segmentation and Product Steering},
    author = {Terstiege, Stefan and Vigier, Adrien},
    journal = {Working Paper},
    year = {2023}
}

@book{topkis,
    author = {Topkis, Donald M.},
    title = {Supermodularity and Complementarity},
    publisher = {Princeton University Press},
    year = {1998}
}

@article{varian,
    author = {Varian, Hal R.},
    title = {Price Discrimination and Social Welfare},
    journal = {American Economic Review},
    volume = {75},
    number = {4},
    pages = {870--875},
    year = {1985}
}

@article{wallheimer18are,
  title        = {Are You Ready for Personalized Pricing?},
  author       = {Wallheimer, Brian},
  journal      = {Chicago Booth Review},
  year         = {2018},
  month        = feb,
  url          = {https://www.chicagobooth.edu/review/are-you-ready-for-personalized-pricing},
}

@article{weitzman1977,
    title = {Is the Price System or Rationing More Effective in Getting a Commodity to Those Who Need It Most?},
    author = {Weitzman, Martin L.},
    journal = {The Bell Journal of Economics},
    pages = {517--524},
    year = {1977}
}

@techreport{whitehouse2015,
    author = {{White House Council of Economic Advisers}},
    title = {Big Data and Differential Pricing},
    year = {2015}
}

@article{wu,
    title = {Consumer Responses to Price Discrimination: Discriminating Bases, Inequality Status, and Information Disclosure Timing Influences},
    journal = {Journal of Business Research},
    volume = {65},
    number = {1},
    pages = {106-116},
    year = {2012},
    author = {Wu, Chi-Cheng and Liu, Yi-Fen and Chen, Ying-Ju and Wang, Chih-Jen}
}

@article{yang2022,
    title = {Selling Consumer Data for Profit: Optimal Market-Segmentation Design and Its Consequences},
    author = {Yang, Kai Hao},
    journal = {American Economic Review},
    volume = {112},
    number = {4},
    pages = {1364--1393},
    year = {2022}
}

@article{steinberg,
title = {Nonprofits with distributional objectives: price discrimination and corner solutions},
journal = {Journal of Public Economics},
volume = {89},
number = {11},
pages = {2205-2230},
year = {2005},
issn = {0047-2727},
doi = {https://doi.org/10.1016/j.jpubeco.2004.08.006},
url = {https://www.sciencedirect.com/science/article/pii/S0047272704001720},
author = {Richard Steinberg and Burton A. Weisbrod}
}

@book{Carlier2022,
  title={Classical and Modern Optimization},
  author={Carlier, Guillaume},
  year={2022},
  publisher={World Scientific}
}

@article{Doligalski2025,
    author = {Doligalski, Pawe\l and Dworczak, Piotr and Akbarpour, Mohammad and Kominers, Scott Duke},
    title = {Optimal Redistribution via Income Taxation and Market Design} ,
    journal = {Working Paper},
    year = {2025},
    nameorder = {random}
}

@article{Ahlvik2025,
  title={Pigouvian Income Taxation},
  author={Ahlvik, Lassi and Liski, Matti and M{\"a}kimattila, Mikael},
  year={2025},
  journal={Working Paper}
}

@article{Salas2025,
    author = {Salas, Aurélien},
    title = {Who is left to pool? Price Discrimination with Endogenous Participation},
    journal = {Working Paper},
    year = {2025}
}

@article{Nikzad2021,
  title={Persuading a pessimist: Simplicity and robustness},
  author={Nikzad, Afshin},
  journal={Games and Economic Behavior},
  volume={129},
  pages={144--157},
  year={2021},
  publisher={Elsevier}
}

@unpublished{terstiege26optimal,
  title        = {Optimal Mechanisms for Selling Market Segmentation},
  author       = {Terstiege, Stefan},
  year         = {2026},
  note         = {Manuscript, February 25, 2026}
}

\end{document}